\pdfoutput=1
\PassOptionsToPackage{table}{xcolor}
\documentclass[acmsmall,nonacm]{acmart}

\usepackage{stmaryrd}
\setcitestyle{numbers,sort&compress}

\AtEndPreamble{%
\theoremstyle{acmdefinition}
}

\usepackage{amssymb}
\usepackage{amsthm}
\usepackage{enumitem}
\usepackage{xspace}
\newtheorem{theorem}{Theorem}[section]

\newtheorem{lemma}[theorem]{Lemma}
\theoremstyle{definition}
\newtheorem{example}[theorem]{Example}

\usepackage{adjustbox}

\usepackage[binary-units]{siunitx}
\DeclareSIUnit{\txn}{txn}
\DeclareSIUnit{\batch}{batch}
\graphicspath{{./DrawIO/}}

\newcommand{\Replica}[1][r]{\MakeUppercase{#1}}

\newcommand{\ID}[1]{\mathop{\textsf{id}}(#1)}

\newcommand{\n}{\mathbf{n}}
\newcommand{\f}{\mathbf{f}}

\newcommand{\Primary}[1]{\mathcal{L}_{#1}}

\newcommand{\SignMessage}[2]{\langle#1\rangle_{#2}}

\newcommand{\Message}[2]{\textsc{#1},#2}

\newcommand{\Name}[1]{\textnormal{\textsc{#1}}}
\newcommand{\BFT}{BFT}
\newcommand{\PBFT}{\Name{PBFT}}

\newcommand{\HS}{\Name{HotStuff}}
\newcommand{\HSTwo}{\Name{HotStuff-2}}
\newcommand{\HSOne}{\Name{HotStuff-1}}
\newcommand{\sysname}{\Name{HotStuff-1}}

\newcommand{\SHSOne}{\Name{Streamlined HotStuff-1}}
\newcommand{\BHSOne}{\Name{Basic HotStuff-1}}
\newcommand{\BHSTwo}{\Name{Basic HotStuff-2}}
\newcommand{\PHSOne}{\Name{Streamlined HotStuff-1}}

\newcommand{\RDB}{\Name{Apache ResilientDB}}

\newcommand{\early}{early finality confirmation}
\newcommand{\strongspec}{Prefix Speculation}

\newcommand{\Certificate}[1]{\mathcal{P}(#1)}
\newcommand{\CCertificate}[1]{\mathcal{C}(#1)}
\newcommand{\Pending}{\mathcal{T}}

\newcommand{\abs}[1]{\lvert #1 \rvert}

\newcommand{\MName}[1]{\textsc{#1}}
\newcommand{\PP}{\mathbb{P}}

\newcommand{\Share}[2]{\delta_{#1}^{\mathcal{#2}}}

\usepackage[noend]{algorithmic}
\usepackage{algorithm,float}


\usepackage{algorithmic}

\newenvironment{myprotocol}{
    \hrule
    \footnotesize
    \smallskip
    \algsetup{linenosize=\footnotesize}
    \begin{algorithmic}[1]
        
        \newcommand{\SPACE}{\item[]}
        \newcommand{\TITLE}[2]{\item[] \textbf{\underline{##1}} (##2) \textbf{:}\\[2pt]}
        \makeatletter
            \newcommand{\EVENT}[1]{\STATE \textbf{event} ##1 \textbf{do}\begin{ALC@g}}
            \newcommand{\ENDEVENT}{\end{ALC@g}}
        \makeatother
        
        \makeatletter
            \newcommand{\FUNCTION}[2]{\STATE \textbf{function} \Name{##1}(##2) \textbf{do}\begin{ALC@g}}
            \newcommand{\ENDFUNCTION}{\end{ALC@g}}
        \makeatother
}{
    \end{algorithmic}%
    \hrule
}

\usepackage{tikz,pgfplots,pgfplotstable}

\usetikzlibrary{shapes,arrows,automata,positioning,cd}
\tikzset{
  spotlessedge/.style   = {black, ->, >=stealth},
}
\usepackage{xcolor}

\usetikzlibrary{arrows.meta,decorations.pathreplacing}
\tikzset{
    >=Stealth,
    smalltext/.append style={scale=0.7},
    dot/.style={circle,scale=0.35,draw=black,fill=black}
}

\usetikzlibrary{automata, positioning, arrows}
\usetikzlibrary{arrows.meta}
\tikzset{
    >=Stealth,
    plot/.append style={baseline,scale=0.475},
    label/.append style={font=\strut\footnotesize},
    dot/.style={circle,scale=0.35,draw=black,fill=black},
}
\pgfplotscreateplotcyclelist{mycyclelist}{
    black   ,every mark/.append style={solid,fill=\pgfplotsmarklistfill},mark=*\\ 
    red     ,every mark/.append style={solid,fill=\pgfplotsmarklistfill},mark=square*\\
    blue    ,every mark/.append style={solid,fill=\pgfplotsmarklistfill},mark=triangle*\\
    teal    ,every mark/.append style={solid,fill=\pgfplotsmarklistfill},mark=pentagon*\\
    brown   ,every mark/.append style={solid,fill=\pgfplotsmarklistfill},mark=halfsquare*\\ 
    orange  ,every mark/.append style={solid,fill=\pgfplotsmarklistfill},mark=halfcircle*\\
    violet  ,every mark/.append style={solid,fill=\pgfplotsmarklistfill,rotate=180},mark=halfdiamond*\\
}
\pgfplotscreateplotcyclelist{mycyclelistex}{
    black   ,every mark/.append style={solid,fill=\pgfplotsmarklistfill},mark=*\\ 
    red     ,every mark/.append style={solid,fill=\pgfplotsmarklistfill},mark=square*\\
}
\pgfplotsset{
    tick label style={font=\large},
    legend style={font=\Large,cells={anchor=west}},
    title style={font=\Large},
    label style={font=\Large},
    width=262.5pt,
    height=185pt,
    every axis/.append style={
        ylabel near ticks,
        xlabel near ticks,
        mark size=2.5pt,
        cycle list name=mycyclelist,
        font=\Large,
        y tick label style={
                        /pgf/number format/precision=1,
                        /pgf/number format/fixed,
                        /pgf/number format/fixed zerofill
                    }
    },
    barstyle/.append style={
        ybar,
        bar width={0.5cm},
        enlarge x limits=0.4,
        enlarge y limits={upper=0.025},
        ymin=0,
        xtick=data
    }
}

\usepackage{tikz,pgfplots,pgfplotstable}
\usetikzlibrary{arrows.meta}
\tikzset{
    >=Stealth,
    node_text/.append style={font=\strut\bfseries},
    plot/.append style={baseline,scale=0.8},
}
\pgfplotscreateplotcyclelist{mycyclelist}{
    thick,black   ,every mark/.append style={solid,fill=\pgfplotsmarklistfill},mark=*\\
    thick,red     ,every mark/.append style={solid,fill=\pgfplotsmarklistfill},mark=square*\\
    thick,blue    ,every mark/.append style={solid,fill=\pgfplotsmarklistfill},mark=triangle*\\
    thick,teal    ,every mark/.append style={solid,fill=\pgfplotsmarklistfill},mark=pentagon*\\
    thick,brown   ,every mark/.append style={solid,fill=\pgfplotsmarklistfill},mark=halfsquare*\\
    thick,orange  ,every mark/.append style={solid,fill=\pgfplotsmarklistfill},mark=halfcircle*\\
    thick,violet  ,every mark/.append style={solid,fill=\pgfplotsmarklistfill,rotate=180},mark=halfdiamond*\\
}
\pgfplotscreateplotcyclelist{mycyclelistex}{
    black   ,every mark/.append style={solid,fill=\pgfplotsmarklistfill},mark=*\\
    red     ,every mark/.append style={solid,fill=\pgfplotsmarklistfill},mark=square*\\
}

\pgfplotsset{
 tick label style={font=\Large},
 legend style={font=\LARGE,cells={anchor=west}},
 title style={font=\Large},
 label style={font=\LARGE},
 width=262.5pt,
 height=185pt,
every axis/.append style={
        xlabel near ticks,
        mark size=2.5pt,
        cycle list name=mycyclelist,
        y tick label style={
            }
    },
barstyle/.append style={
        ybar,
       bar width={0.5cm},
        enlarge x limits=0.4,
        enlarge y limits={upper=0.025},
        ymin=0,
        xtick=data
    }
every axis/.append style={
ylabel near ticks,
xlabel near ticks,
mark size=1pt,
cycle list name=mycyclelist,
 font=\Large,
enlargelimits=0.1
},
}

\pgfplotsset{select coords index/.style n args={2}{
    x filter/.code={
        \ifnum\coordindex<#1\def\pgfmathresult{}\fi
        \ifnum\coordindex>#2\def\pgfmathresult{}\fi
    }
}}

\definecolor{colAA}{RGB}{255,87,51}
\definecolor{colAB}{RGB}{255,116,51}
\definecolor{colBA}{RGB}{53,142,250}
\definecolor{colBB}{RGB}{147,199,251}
\definecolor{colCA}{RGB}{92,208,48}
\definecolor{colCB}{RGB}{167,252,105}
\definecolor{colDA}{RGB}{214,217,60}
\definecolor{colDB}{RGB}{242,245,52}
\definecolor{colEA}{RGB}{242,84,228}
\definecolor{colEB}{RGB}{223,128,216}
\definecolor{colFA}{RGB}{159,28,249}
\definecolor{colGA}{RGB}{28,71,249}

\pgfplotstableread{
n	HS	HS-2	HS-1	Snake
4	89945	90253	90286	90152
16	65173	65236	64815	64928
32	55642	55470	55303	56983
64	37627	37479	37314	37711
}\dataTputb

\pgfplotstableread{
n HS	HS-2	HS-1	Snake
4	5.4735	4.3366	3.2029	3.21112
16	7.8041	5.9582	4.5172	4.4728
32	9.0037	7.171	5.5422	5.336
64	14.0277	11.2806	8.4171	8.4536
}\dataClientLat

\pgfplotstableread{
n HS	HS-2	HS-1	Snake
4	3.93202	2.80528	1.67575	1.6802
16	5.6044	4.03332	2.40019	2.3925
32	6.6613	4.83231	3.03961	2.922
64	10.3404	7.57379	4.70942	4.8042
}\dataConsensusLat

\pgfplotstableread{
n	HS	HS-2	HS-1	Snake
100	55642	55470	55303	56983
1000 393056	397826	397826	394827
2000 653600	651938	657016	655849
5000 1007000 1003582 1003935 1001097
10000 998000 998000	998000	1000000
}\dataBatchSize

\pgfplotstableread{
n	HS	HS-2	HS-1	Snake
100	9.0037	7.171	5.5422	5.336
1000	20.1699	15.82379	11.11878	11.3962
2000	30.6009	25.94823	17.11501	17.0732
5000	59.5984	47.59186	32.2121	31.4161
10000	110.4805 84.84423	60.077	61.8673
}\ClientLatBatchSize

\pgfplotstableread{
n	HS	HS-2	HS-1	Snake
100	6.6613	4.83231	3.03961	2.922
1000	10.0267	7.83271	4.7651	4.7077
2000	15.9466	13.55792	8.6231	8.5228
5000	32.0573	25.97071	17.1468 17.3619
10000	59.7391	46.78255	31.9661	32.1778
}\ConsensusLatBatchSize

\newcommand{\throughputgraphb}[5]{\begin{tikzpicture}[plot, scale=1.2]
    \begin{axis}[yticklabel style={/pgf/number format/precision=0},xlabel={#3},ylabel={#4},title={#2},xtick={#5}, ymin={1},legend columns=-1, legend to name={mainlegend4}, 
        grid=major, xmin={0}, width=200, height=140]
        \addplot table[x={n},y={HS}] {#1};
        \addplot table[x={n},y={HS-2}] {#1};
        \addplot table[x={n},y={HS-1}] {#1};
        \addplot table[x={n},y={Snake}] {#1};
        \legend{HotStuff, HS-2, \sysname{} , \sysname{} (with slotting)}
    \end{axis}
\end{tikzpicture}}

\newcommand{\geothroughputgraph}[5]{\begin{tikzpicture}[plot, scale=1.2]
    \begin{axis}[yticklabel style={/pgf/number format/precision=0},xlabel={#3},ylabel={#4},title={#2},xtick={#5}, ymin={1},legend columns=-1, legend to name={mainlegend4}, 
        grid=major, xmin={1.5}, width=200, height=140]
        \addplot table[x={n},y={HS}] {#1};
        \addplot table[x={n},y={HS-2}] {#1};
        \addplot table[x={n},y={HS-1}] {#1};
        \addplot table[x={n},y={Snake}] {#1};
        \legend{HotStuff, HotStuff-2, \sysname{} , \sysname{} (with slotting)}
    \end{axis}
\end{tikzpicture}}

\newcommand{\latencygraph}[5]{\begin{tikzpicture}[plot, scale=1.2]
    \begin{axis}[yticklabel style={/pgf/number format/precision=0}, xlabel={#3},ylabel={#4},title={#2},xtick={#5}, ymin={0},legend columns=-1, ymax={20},
        grid=major, xmin={0}, width=200, height=140]
        \addplot table[x={n},y={HS}] {#1};
        \addplot table[x={n},y={HS-2}] {#1};
        \addplot table[x={n},y={HS-1}] {#1};
        \addplot table[x={n},y={Snake}] {#1};
    \end{axis}
\end{tikzpicture}}

\newcommand{\geolatencygraph}[5]{\begin{tikzpicture}[plot, scale=1.2]
    \begin{axis}[yticklabel style={/pgf/number format/precision=0}, xlabel={#3},ylabel={#4},title={#2},xtick={#5}, ymin={0},legend columns=-1, ymax={4},
        grid=major, xmin={1.5}, width=200, height=140]
        \addplot table[x={n},y={HS}] {#1};
        \addplot table[x={n},y={HS-2}] {#1};
        \addplot table[x={n},y={HS-1}] {#1};
        \addplot table[x={n},y={Snake}] {#1};
    \end{axis}
\end{tikzpicture}}

\newcommand{\nonresponsivegraph}[5]{\begin{tikzpicture}[plot, scale=1.2]
    \begin{axis}[yticklabel style={/pgf/number format/precision=0},xlabel={#3},ylabel={#4},title={#2},xtick={#5}, ymin={0},legend columns=-1, 
        grid=major,  xmin={-1}, width=200, height=140]
        \addplot table[x={n},y={HS}] {#1};
        \addplot table[x={n},y={HS-2}] {#1};
        \addplot table[x={n},y={HS-1}] {#1};
        \addplot table[x={n},y={Snake}] {#1};
    \end{axis}
\end{tikzpicture}}

\newcommand{\nonresponsivegraphb}[5]{\begin{tikzpicture}[plot, scale=1.2]
    \begin{axis}[yticklabel style={/pgf/number format/precision=0},xlabel={#3},ylabel={#4},title={#2},xtick={#5}, ymin={0},legend columns=-1, 
        grid=major,  xmin={-1}, width=200, height=140]
        \addplot table[x={n},y={HS}] {#1};
        \addplot table[x={n},y={HS-2}] {#1};
        \addplot table[x={n},y={HS-1}] {#1};
        \addplot table[x={n},y={Empty}] {#1};
        \addplot table[x={n},y={Snake}] {#1};
    \end{axis}
\end{tikzpicture}}

\newcommand{\nonresponsivelatgraph}[5]{\begin{tikzpicture}[plot, scale=1.2]
    \begin{axis}[yticklabel style={/pgf/number format/precision=0},xlabel={#3},ylabel={#4},title={#2},xtick={#5}, ymin={0},legend columns=-1, ymax={200},
        grid=major,  xmin={-1}, width=200, height=140]
        \addplot table[x={n},y={HS}] {#1};
        \addplot table[x={n},y={HS-2}] {#1};
        \addplot table[x={n},y={HS-1}] {#1};
        \addplot table[x={n},y={Snake}] {#1};
    \end{axis}
\end{tikzpicture}}

\newcommand{\nonresponsivelatgraphb}[5]{\begin{tikzpicture}[plot, scale=1.2]
    \begin{axis}[yticklabel style={/pgf/number format/precision=0},xlabel={#3},ylabel={#4},title={#2},xtick={#5}, ymin={0},legend columns=-1, ymax={200},
        grid=major,  xmin={-1}, width=200, height=140]
        \addplot table[x={n},y={HS}] {#1};
        \addplot table[x={n},y={HS-2}] {#1};
        \addplot table[x={n},y={HS-1}] {#1};
        \addplot table[x={n},y={Empty}] {#1};
        \addplot table[x={n},y={Snake}] {#1};
    \end{axis}
\end{tikzpicture}}

\newcommand{\forktailgraph}[5]{\begin{tikzpicture}[plot, scale=1.2]
    \begin{axis}[yticklabel style={/pgf/number format/precision=0},xlabel={#3},ylabel={#4},title={#2},xtick={#5}, ymin={0},legend columns=-1, 
        grid=major,  xmin={-1}, width=200, height=140,legend to name={falegend}]
        \addplot table[x={n},y={HS}] {#1};
        \addplot table[x={n},y={HS-2}] {#1};
        \addplot table[x={n},y={HS-1}] {#1};
        \addplot table[x={n},y={Snake10}] {#1};
        \addplot table[x={n},y={Snake100}] {#1};
        \addplot table[x={n},y={Zyzzyva}] {#1};
        \legend{HotStuff, HotStuff-2, HotStuff-1, HotStuff-1 (10ms-slotting), HotStuff-1 (100ms-slotting)}
    \end{axis}
\end{tikzpicture}}

\newcommand{\batchsizegraph}[5]{\begin{tikzpicture}[plot, scale=1.2]
    \begin{axis}[yticklabel style={/pgf/number format/precision=1}, xlabel={#3},ylabel={#4},title={#2},symbolic x coords={#5}, ymin={0},legend columns=-1,
        grid=major, width=200, height=140]
        \addplot table[x={n},y={HS}] {#1};
        \addplot table[x={n},y={HS-2}] {#1};
        \addplot table[x={n},y={HS-1}] {#1};
        \addplot table[x={n},y={Snake}] {#1};
    \end{axis}
\end{tikzpicture}}

\newcommand{\batchsizeLatgraph}[5]{\begin{tikzpicture}[plot, scale=1.2]
    \begin{axis}[yticklabel style={/pgf/number format/precision=0}, xlabel={#3},ylabel={#4},title={#2},symbolic x coords={#5}, ymin={0},legend columns=-1, ymax={120},
        grid=major, width=200, height=140]
        \addplot table[x={n},y={HS}] {#1};
        \addplot table[x={n},y={HS-2}] {#1};
        \addplot table[x={n},y={HS-1}] {#1};
        \addplot table[x={n},y={Snake}] {#1};
    \end{axis}
\end{tikzpicture}}

\newcommand{\networkdelaytesttput}[4]{%
\begin{tikzpicture}[plot, scale=1]
    \begin{axis}[
        title={#2},
        xlabel={#3},
        ylabel={#4},
        xtick={0, 10, 11, 20, 21, 31},
        xticklabels={0, $\,\,10\,\,11$, , $\,\,20\,\,21$, ,31}, 
        ymode=log, 
        log basis y=10, 
        grid=major,
        scaled x ticks=false, 
        width=190,
        height=140,
        legend columns=-1,
        legend to name={networkdelaylegend},
        scaled y ticks=false,
        tick label style={/pgf/number format/.cd, fixed, precision=1}, 
        yticklabel style={/pgf/number format/.cd, fixed, precision=1}, 
        ytick scale label code/.code={%
            \pgfmathparse{int(#1)} 
            \pgfmathprintnumber[fixed, precision=1]{\pgfmathresult}
        }
    ]
        \addplot table[x={n},y={HS}] {#1};
        \addplot table[x={n},y={HS-2}] {#1};
        \addplot table[x={n},y={HS-1}] {#1};
        \addplot table[x={n},y={Snake}] {#1};
        \legend{HotStuff, HS-2, \sysname{}, \sysname{} (with slotting)}
    \end{axis}
\end{tikzpicture}%
}

\newcommand{\networkdelaytestlat}[4]{%
\begin{tikzpicture}[plot, scale=1]
    \begin{axis}[
        title={#2},
        xlabel={#3},
        ylabel={#4},
        xtick={0, 10, 11, 20, 21, 31},
        xticklabels={0, $\,\,10\,\,11$, , $\,\,20\,\,21$, ,31}, 
        grid=major,
        scaled x ticks=false, 
        width=200,
        height=140,
        legend columns=-1
    ]
        \addplot table[x={n},y={HS}] {#1};
        \addplot table[x={n},y={HS-2}] {#1};
        \addplot table[x={n},y={HS-1}] {#1};
        \addplot table[x={n},y={Snake}] {#1};
    \end{axis}
\end{tikzpicture}%
}

\pgfplotstableread{
n	HS-1	HS-2	HS	Snake10 Snake100	Zyzzyva
0	56245	55407	55303	56983 -1000	152680
1	49212	49945	50466	54977 -1000	54345
4	31346	32148	31848	46298 -1000	54492
7	20543	20076	21770	45532 -1000	54374
10	15202	14461	15982	39932 -1000	54867
}\datacrashtenms

\pgfplotstableread{
n	HS-1	HS-2	HS	Snake10 Snake100	Zyzzyva
0	56245	55407	55303  -1000	56983	192918
1	20433	20731	20387  -1000	53790	116030
4	6158	6030	6161  -1000	48538	118838
7	3324	3357	3276  -1000	42923	125781
10	2104	2112	2200  -1000	36705	124212
}\datacrashhundredms

\pgfplotstableread{
n HS-1	HS-2	HS	Snake10 Snake100	Zyzzyva
0	5.5422	7.171	9.0037	5.336 -1000	    2.85475
1	6.603	8.134	11.861	5.2078 -1000	9.0054
4	11.754	12.5258	30.917	5.6563 -1000	8.9384
7	21.6921	19.872	69.0802	5.6839 -1000	8.9847
10	25.9234	26.9076	135.92	6.1177 -1000	8.98074
}\latcrashtenms

\pgfplotstableread{
n HS-1	HS-2	HS	Snake10 Snake100	Zyzzyva
0	5.5422	7.171	9.0037	 -1000	5.336	13.1798
1	19.4218	21.3076	33.708  -1000	5.4213	20.068
4	65.035	68.638	194.13  -1000	5.6162	20.1787
7	122.957	122.926	553.424  -1000	5.9481	20.09187
10	194.805 193.815	1162.23  -1000	6.5191	20.1008
}\latcrashhundredms

\pgfplotstableread{
n HS HS-2 HS-1 Snake
0 55642 55470	55303 55983
1 47048 47252 47301 54977
4 32665 32642 32537 48998
7 24479 24938 24916 45032
10 20162 20196 20172 39932
}\dataNR

\pgfplotstableread{
n HS HS-2 HS-1 Snake Empty
0 55642 55470 55303 55983 -100
1 20264 20229 20222 53790 -100
4 6962 6939 6975 48538 -100
7 4189 4132 4211 42923 -100
10 3015 3019 3017 36705 -100
}\dataNRB

\pgfplotstableread{
n HS HS-2 HS-1 Snake
0 9.0483  7.1089 5.3013	5.161
1 10.6237 8.4064 6.2944 5.2078
4 15.4016 12.2794 9.2198 5.6563
7 20.6294 16.1366 12.0650 5.6839
10 24.9816 19.9252 14.9602 6.1177
}\latNR

\pgfplotstableread{
n HS HS-2 HS-1 Snake Empty
0 9.0483 7.1089	5.3013	5.129 -100
1 24.812 19.8532 14.8907 5.4213 -100
4 72.384 58.2012 43.9212 5.6162 -100
7 124.2736 98.8459 72.5262 5.9481 -100
10 168.3731 134.4022 100.9351 6.5191 -100
}\latNRB

\pgfplotstableread{
n	HS	HS-2	HS-1	Snake10 Snake100 Zyzzyva
0	55114	55268	55547	56410	56483 -1000
1	53475	53272	53312	56358	56622 -1000
4	48106	47367	48249	55598	56266 -1000
7	43448	43280	43103	54981	56033 -1000
10	37906	38058	38006	54079	55691 -1000
}\dataFA

\pgfplotstableread{
n	HS	HS-2	HS-1	Snake10 Snake100 Zyzzyva
0	9.0483	7.1089	5.3013	5.161	5.129 -1000
1	9.3233	7.4815	5.5567	5.1330	5.151 -1000
4	10.3792	8.4208	6.1187	5.1369	5.172 -1000
7	11.444	9.1474	6.8448	5.1432	5.166 -1000
10	13.145	10.3892	7.7204	5.1583	5.148 -1000
}\latFA

\pgfplotstableread{
n	HS	HS-2	HS-1	Snake10 Snake100 Zyzzyva
0 -1000 -1000 54547	55410	55829    -1000
1 -1000 -1000 52556	54241	54949 -1000
4 -1000 -1000 47341	55058	55392 -1000
7 -1000 -1000 41321	54742	55595 -1000
10 -1000 -1000 33775	54275	55186 -1000
}\dataRB

\pgfplotstableread{
n	HS	HS-2	HS-1	Snake10 Snake100 Zyzzyva
0 -100 -100	5.6013	5.261	5.42437 -1000
1 -100 -100	5.8114	5.288	5.39282 -1000
4 -100 -100	5.9248	5.233	5.4476 -1000
7 -100 -100	6.2697	5.297	5.43749 -1000
10 -100 -100 7.608	5.3507	5.4837 -1000
}\latRB

\pgfplotstableread{
n	HS	HS-2	HS-1	Snake
2	449	453	449	444
3	275	270	278	274
4	191	193	192	190
5	187	186	184	182
}\datageo

\pgfplotstableread{
n HS	HS-2	HS-1	Snake
2	452	451	447	444
3	298	291	286	286
4	186	187	188	188
5	177	178	176	172
}\datageotpcc

\pgfplotstableread{
n	HS	HS-2	HS-1	Snake
2	1.1979	0.94802	0.719953	0.720343
3	2.08495	1.6589	1.240168	1.24567
4	2.87071	2.26848	1.711321	1.7173
5	3.10781	2.46017	1.840360	1.847634
}\responselatgeotpcc

\pgfplotstableread{
n	HS	HS-2	HS-1	Snake
2	1.21723	0.95258	0.72268	0.73389
3	2.22716	1.78365	1.28577	1.33495
4	2.78928	2.20764	1.68553	1.70249
5	2.80243	2.27834	1.77226	1.80982
}\responselatgeo

\pgfplotstableread{
n	HS	HS-2	HS-1	Snake
2	0.78624	0.53911	0.304304	0.301984
3	1.50332	1.06791	0.617453	0.614204
4	1.90528	1.34122	0.792393	0.781203
5	1.88358 1.34527	0.801935	0.803262
}\consensuslatgeo

\pgfplotstableread{
n  HS-1   HS-2   HS   Snake
0  61413  61480  61434  61540
10 38649  38955  38855  43137
11 28236  28058  28112  30344
20 23156  23626  23404  24337
21 23055  22683  22650  23874
31 20153  19824  19787  20655
}\networkdelayonemstput

\pgfplotstableread{
n   HS-1   HS-2     HS    Snake
0   0.00481437  0.0064456   0.0081136   0.00473137
10  0.00865429  0.0106837   0.0133586   0.007182
11  0.0116237   0.0147005   0.018322    0.009924
20  0.0140165   0.017563    0.022072    0.012425
21  0.014507    0.018449    0.022905    0.0126356
31  0.0168918   0.021979    0.0271209   0.0143684
}\networkdelayonemslat

\pgfplotstableread{
n	HS-1	HS-2	HS	Snake
0	61413	61480	61434	61540
10	26110	26209	26218	38044
11	7790	7778	7789	8045
20	6786	6782	6784	6865
21	6690	6692	6720	6752
31	5914	5913	5917	5940
}\networkdelayfivemstput

\pgfplotstableread{
n	HS-1  HS-2    HS    Snake
0	0.00481437	0.0064456	0.0081136	0.00473137
10	0.012585	0.015829	0.019672	0.0107397
11	0.043433	0.052099	0.065056	0.038668
20	0.0493909	0.060208	0.0749912	0.0449543
21	0.0500853	0.064558	0.081278	0.045586
31	0.057349	0.074019	0.0910562	0.050812
}\networkdelayfivemslat

\pgfplotstableread{
n	HS-1	HS-2	HS	Snake
0	61413	61480	61434	61540
10	4279	4281	4272	21291
11	840	839	840	855
20	751	751	751	765
21	738	738	738	755
31	656	656	657	654
}\networkdelayfiftymstput

\pgfplotstableread{
n	HS-1  HS-2    HS    Snake
0	0.00481437	0.0064456	0.0081136	0.00473137
10	0.072883	0.0957783	0.119966	0.0503423
11	0.406511	0.478674	0.59903	    0.36829
20	0.451494	0.53998	    0.676957	0.40563
21	0.458069	0.595957	0.735043	0.409172
31	0.5132235	0.666422	0.8231625	0.4572
}\networkdelayfiftymslat

\pgfplotstableread{
n	HS-1	HS-2	HS	Snake
0	61413	61480	61434	61540
10	447	446	446	14609
11	83	84	84	84
20	73	74	73	76
21	72	74	73	75
31	65	66	66	66
}\networkdelayfivehundredmstput

\pgfplotstableread{
n  HS-1   HS-2   HS   Snake
0	0.004814	0.0064456	0.008113	0.00473137
10	0.966892	1.27789	    1.46878	    0.534457
11	4.026693	4.95771	    6.297183	3.52044
20	4.405486	5.527367	6.9454675	3.93304
21	4.43314	    5.99258	    7.387172	3.980818
31	4.830668	6.615701	8.054873	4.349064
}\networkdelayfivehundredmslat

\pgfplotstableread{
n	HS-1	HS-2	HS	Snake
0	61413	61480	61434	61540
10	1304	1301	1309	35080
11	456	456	456	518
20	450	450	450	491
21	657	680	681	871
31	869	876	876	1020
}\tworegiontput

\pgfplotstableread{
n	HS-1  HS-2    HS    Snake
0	0.00481437	0.0064456	0.0081136	0.00473137
10	0.2181127	0.2911497	0.3603013	0.174266
11	0.6832123	0.901414	1.128221	0.611053
20	0.6865525	0.9084739	1.142802	0.624042
21	0.494425	0.622578	0.783592	0.41725
31	0.356363	0.4723919	0.580161	0.292626
}\tworegionlat

\newcommand{\axisnodes}{Number of replicas ($\n$)}
\newcommand{\axisregions}{Number of regions}

\newcommand{\axisbatches}{Batch size}

\newcommand{\axisticksregions}{2, 3, 4, 5}

\newcommand{\axistickslows}{0,1,4,7,10}

\newcommand{\axisticksnodesb}{4,16,32,64}

\newcommand{\DM}[1]{{\color{orange}{Dahlia: {#1}}}}

\newcommand{\Arxiv}[1]{#1}

\newcommand{\Add}[1]{#1}

\title{\textsc{HotStuff-1}: Linear Consensus with One-Phase Speculation}

\author{Dakai Kang}
\email{dakang@ucdavis.edu}
\orcid{0002-7867-3681}
\affiliation{%
  \institution{University of California, Davis}
  \country{USA}
}

\author{Suyash Gupta}
\email{suyash@uoregon.edu}
\orcid{0002-3240-1840}
\affiliation{%
  \institution{University of Oregon}
  \country{USA}
}

\author{Dahlia Malkhi}
\email{dahliamalkhi@ucsb.edu}
\orcid{0002-7038-7250}
\affiliation{%
  \institution{University of California, Santa Barbara}
  \country{USA}
}

\author{Mohammad Sadoghi}
\email{msadoghi@ucdavis.edu}
\orcid{0003-2779-6080}
\affiliation{%
  \institution{University of California, Davis}
  \country{USA}
}

\begin{document}

\begin{abstract}
This paper introduces \HSOne, a BFT consensus protocol that improves the latency of \HSTwo{} by two network hops while maintaining linear communication complexity against faults. 
Furthermore, \HSOne{} incorporates an incentive-compatible leader rotation design that motivates leaders to propose transactions promptly.
\HSOne{} achieves a reduction of two network hops by \textit{speculatively} sending clients \early{}s, after one phase of the protocol.
Introducing speculation into streamlined protocols is challenging because, unlike stable-leader protocols, these protocols cannot stop the consensus and recover from failures.
Thus, we identify {\em prefix speculation dilemma} in the context of streamlined protocols; 
\HSOne{} is the first streamlined protocol to resolve it.
\HSOne{} embodies an additional mechanism, \textit{slotting}, that thwarts delays caused by (1) rationally-incentivized leaders and
(2) malicious leaders inclined to sabotage other's progress. 
The slotting mechanism allows leaders to dynamically drive as many decisions as allowed by network transmission delays before view timers expire, thus mitigating both threats.



\end{abstract}




\maketitle

\section{Introduction}
\label{sec:intro}
This paper introduces \sysname{}, a \BFT{} consensus protocol designed to reduce latency while simultaneously maintaining scalability.
\sysname{} is primarily motivated by blockchains and online platforms that support digital asset payments and marketplaces~\cite{bitcoin,ethereum,hyperledger}. 
These systems employ a \BFT{} consensus protocol because it enables 
them to provide their clients access to a verifiable immutable ledger managed by multiple 
distrusting nodes, some of which may be malicious.
In these systems, especially financial platforms, response latency is crucial for user engagement and satisfaction. 
Moreover, the demands for low response latency are posed not only by the market but also by regulation. 
A manuscript detailing regulatory technical requirements for Financial Market Infrastructure (FMI) states $12$ key standards for 
operating an FMI, among which are performance requirements such as meeting peak throughput demand and timely responsiveness~\cite{bis-pfmi}.

In this paper, we are 
interested in \BFT{} consensus protocols for a {\em partially-synchronous} setting, due to 
their safety against temporary network delays.  
Pioneering \BFT{} consensus protocols belonging to the \PBFT{} family~\cite{pbftj,blockchain-book}
employ a {\em stable-leader} design, where one replica designated as the leader initiates 
a two-phase consensus algorithm that determines the ledger. 
Unfortunately, the stable-leader design has some drawbacks. 

\textbf{D1:} a dedicated leader increases {\em censorship opportunities}, as the leader decides what transactions to propose~\cite{hotstuff}. 

\textbf{D2:} when the leader fails, these protocols 
switch to a {\em view-change} algorithm that incurs quadratic communication complexity to replace the leader (or change the view) and
drops the system throughput to zero, as consensus on new transactions can start only after the view-change~\cite{aardvark,prime}. 

\textbf{D3:} it {\em inhibits load and reward balancing} among the replicas~\cite{rcc}.

\textbf{D4:} a malicious leader can {\em keep the system throughput at the lowest level} and prevent detection by proposing transactions just before the timeout period~\cite{aardvark,prime,rbft}.
 
Some recent protocols that follow the stable-leader design attempt to solve \textbf{D3} and \textbf{D4} by requiring all the replicas to act 
as the leader and/or track the leader's performance~\cite{aardvark,prime,rbft,mirbft,rcc,fairledger}. 
However, these works require several redundant rounds of consensus that track the leader's performance (e.g. RBFT~\cite{rbft} and Fairledger~\cite{fairledger}) and 
face collusion attacks by multiple malicious leaders (e.g. MirBFT~\cite{mirbft} and RCC~\cite{rcc}). 

Alternatives to the stable-leader design emerged in the blockchain world. 
First, Tendermint introduced a design that proactively replaces the leader at the end of each 
consensus decision~\cite{tendermint}. 
Later, HotStuff~\cite{hotstuff} reduced view-change communication costs to linear, and additionally \textit{streamlined} protocol 
phases to (at least) double throughput (solving \textbf{D1} to \textbf{D4}). 
Thus, streamlined linear protocols in the \HS{} 
family mitigate the drop in system throughput by allowing regular leader replacement at 
(essentially) no communication cost. 
However, these protocols face the following three additional challenges: 

{\bf D5: Increased latency.}
Despite recent improvements~\cite{hs2}, streamlined protocols incur higher latency 
than the stable-leader protocols that employ optimizations like {\em speculative-execution}~\cite{sbft,poe}.

{\bf D6: Leader-slowness phenomenon.} 
In blockchain systems, regular leader replacement creates an undesirable incentive structure: 
a leader may be inclined to delay proposing a block of transactions as close as possible to the end of its
view expiration period in order to pick the transactions that offer the highest fees. 
Similarly, block-builders participating in a proposer-builder auction will wait as long as possible to maximize MEV 
(maximal extractable value) exploits~\cite{Daian2019FlashB2,pbs2024,time-role-mev}. 
Thus, rational leaders/builders may slow down progress and cause clients to suffer increased latency.

{\bf D7: Tail-forking attack.}
BeeGees~\cite{beegees} exposed another vulnerability of streamlined protocols, 
where faulty leaders prevent proposals by correct leaders from being committed unless 
there are consecutive correct leaders. This attack surfaces when faulty leaders 
are interjected between correct leaders as leaders are rotated. While they may not succeed 
in completely censoring transactions, faulty leaders may cause specific clients to suffer 
increased latency and overall, slow down progress.

Thus, we are facing a conundrum: 
on the one hand, stable-leader protocols yield optimal latency under no-failure cases through speculative execution and do not face \textbf{D5} to \textbf{D7}.
However, they have yet to solve \textbf{D1} and \textbf{D2}, and solving \textbf{D3} and \textbf{D4} introduces new challenges.
On the other hand, streamlined protocols resolve \textbf{D1} to \textbf{D4} but have yet to solve 
\textbf{D5} to \textbf{D7}. 

\sysname{} resolves these seeming trade-offs by introducing a \BFT{} consensus solution 
that embodies two principal contributions:
\begin{enumerate}
    \item  
    A novel algorithmic core that combines regular leader rotation with linear communication, streamlining and speculative execution. 
    \HSOne{} acts as an optimist by speculatively executing client requests and serving the clients with the results 
    of uncommitted transactions.
    
    \item An {\em adaptive slotting} algorithm that provides each leader with multiple slots 
    to propose transactions. 
    \HSOne{} uses slotting to maintain consistent high performance by mitigating the impacts of leader-slowness and tail-forking.

\end{enumerate}

\paragraph{Early Finality Confirmation through Speculation} 
The notion of applying speculative execution to \BFT{} protocols is not new. 
In his PhD thesis~\cite{miguel-thesis}, Miguel Castro  presented the idea of applying {\em tentative execution} to \PBFT{}, which was later expanded/evaluated by PoE~\cite{poe}.
Several other flavors of speculative execution also exist (Zyzzyva~\cite{zyzzyvaj} and SBFT~\cite{sbft}).
These papers illustrate that speculative execution can reduce the latency of \BFT{} consensus in the no-failure case.
Unfortunately, applying speculative execution to streamlined protocols is not a straightforward extension.
 

These stable-leader protocols {\bf stop} speculative execution during the recovery/view change phases because they need to run an explicit view-change protocol (\textbf{D2}).
At the end of the view-change protocol, all replicas start the new view when they receive from the new leader a {\em state}.
This state starts from the last agreed-upon checkpoint, and for each sequence number that some replica claims to have observed since the last checkpoint, this state includes a prepare-certificate (if available) or a proposal from the previous leader.%
\footnote{Alternatively, if the new leader does not have access to any prepare-certificate for a sequence number, it can leave 
that sequence number as empty~\cite{pbftj}.}
However, before a replica can add any of these sequence numbers/proposals to its log, the leader needs to rerun consensus on each of them.

Streamlined protocols {\bf do not have} the option of stopping consensus and rerunning consensus on past transactions, which makes introducing speculation challenging. 
Thus, we identify the existence of a conundrum when applying speculation to the streamlined protocols; 
we term this conundrum as the {\em prefix speculation dilemma}.
\HSOne{} is the first streamlined protocol to employ speculative execution and resolve this conundrum by dictating when it is safe for a replica to speculatively execute a proposal.

Consequently, \HSOne{} treats clients as first-class citizens of consensus by serving them with \textit{\early}. 
\HSOne{} builds streamlining and speculation over \HSTwo~\cite{hs2}. 
Unlike \HSTwo{}, which forces replicas to wait until they learn whether a transaction has committed, \HSOne{} allows replicas to send commit-votes on transactions directly to clients when a transaction is prepared and highly likely to commit, which also allows replicas to speculate on the execution results and send responses to clients.
On collecting responses from a quorum of $\n-\f$ replicas, clients learn two things at once: a commit decision and its execution result, which enables an \early{}.
Thus, \HSOne{} meets the challenges \textbf{D1} to \textbf{D5}.

\paragraph{Low latency through slotting.}
\HSOne{} resolves a subset of the challenges we listed earlier in this section, 
but challenges like leader slowness (\textbf{D6}) and tail-forking attacks (\textbf{D7}) remain.
Therefore, we incorporate a novel {\em slotting} mechanism into \HSOne{}. 
Slotting allows each leader to propose multiple successive blocks of transactions; each leader has access to multiple slots and can propose one block of transactions per slot. 
Assigning more than one slot to a leader motivates a rational leader 
to ensure that its blocks commit quickly, opening the opportunity to propose more new blocks.
%
However, fixing the number of slots per leader/view does not eliminate the slowness attack; a fast leader will slow down its last slot.
Therefore, we devise an \textit{adaptive} slotting mechanism that allows a leader to propose as many slots 
as it can during the time span allotted to its view.
Permitting adaptive slotting in a streamlined consensus protocol unravels a new challenge:
how can the subsequent leader determine if it has received the certificates corresponding to the last slot of the preceding leader? 
We introduce the notion of \emph{trusted/distrusted previous leaders} to enable a correct leader to propose its first slot at the network speed between itself and the previous leader if the previous leader is correct.

\paragraph{Resilience to tail-forking attacks.} \HSOne{} with slotting guarantees that in each view $v$, if $\Primary{v}$ proposed at least two slots, at most one could remain uncertified, and it could only be tail-forked if fewer than $\f+1$ correct replicas voted for it. This is achieved via \emph{carry} blocks and a dual-certificate mechanism: \emph{New-View} and \emph{New-Slot} certificates, which enforce the slot’s inclusion in the well-formed first-slot proposal sent by the next leader.


We illustrate the practicality of our design by implementing \HSOne{} (with and without slotting) in \RDB{} (incubating)~\cite{apacheresdb} 
and evaluating it against two baselines: \HS{} and \HSTwo{}. 
Our results affirm that \HSOne{} yields lower latency than the baselines; 
in the no-failure case, \HSOne{} (with and without slotting) yields up to $41.5\%$ and $24.2\%$ lower latency.
Additionally, we illustrate the resistance of \HSOne{} (with slotting) against leader-slowness and tail-forking attacks.
In summary, we make the following contributions:
\begin{enumerate}[wide,nosep]
    \item We introduce \HSOne{}, the first speculative, streamlined and linear \BFT{} consensus protocol that serves clients with \early{s} for their transactions.

    \item We expose a prefix speculation dilemma that exists in the context of streamlined \BFT{} protocols that employ speculation 
    and present a solution tailored for \HSOne{}.


    \item We introduce slotting in \HSOne{} to mitigate leader-slowness and tail-forking attacks.
   Our slotting mechanism is adaptive, yet guarantees no delay for subsequent leaders. 
    
\end{enumerate}

\section{Background and System Model}
\label{sec:model}
Modern databases require replication to guarantee availability to their clients; 
consensus protocols help keep these replicas consistent~\cite{paxossimple,raft}.
As consensus can quickly bottleneck system performance, a large body of existing work attempts to
optimize these protocols~\cite{distbook}.
A majority of existing databases~\cite{spanner,tidb,cockroachdb} employ {\em crash fault-tolerant} 
consensus protocols to guarantee consistent replication despite crash failures~\cite{paxossimple,raft}.
This crash-failure threat model is sufficient for these databases as they are managed by a single organization.
In this paper, we focus on designing efficient Byzantine Fault-Tolerant (\BFT{}) consensus protocols that can guard against arbitrary malicious failures. Such protocols are necessary for databases managed by multiple parties; commonly used in financial trading and blockchain applications~\cite{apacheresdb,caper,basil}.

We assume the system model adopted by existing partially synchronous \BFT{} 
consensus protocols~\cite{pbftj,zyzzyvaj,hotstuff,poe,sbft}. We assume a system 
of $\n$ replicas, of which at most $\f$ are faulty (malicious or crash-failed), and 
the remaining $\n-\f$ replicas are correct; $\n \geq 3\f + 1$. Correct replicas 
follow the protocol: on the same input, produce the same output. This system receives 
requests from a set of clients; any number of clients can be faulty. We use $R$ 
and $c$ to denote a replica and a client, and each replica is assigned a unique 
identifier in the range $[1,\n]$ using function $\ID{R}$.

{\bf Authenticated communication}:
each client/replica uses digital signatures to sign a message~\cite{cryptobook}.
Additionally, replicas make use of the BLS threshold signature scheme~\cite{bls} 
to form $(\n,t)$ threshold signatures. Each replica $R$ has access to a private 
signature key, which it uses to create a signature share $\Share{R}{}$. An aggregator 
needs only $t$ shares out of $\n$ to create the threshold signature. A receiver can use 
the corresponding public key to verify whether at least $t$ replicas contributed to 
this signature. We use the notation $\SignMessage{m}{R}$ to denote a signature or a 
threshold signature share on message $m$ by replica $R$. Correct replicas only 
accept {\em well-formed} messages that have a valid signature.
Further, we assume the existence of a collision-resistant hash function $H(x)$, where it is 
impossible to find a value $x'$, such that $H(x) = H(x')$~\cite{cryptobook}.

{\bf Adversary model}:
Faulty replicas can delay, drop, and duplicate any message and 
collude with each other. However, a faulty replica cannot forge the 
identity/messages of a correct replica. 

{\bf Synchrony}:
We assume a partial synchrony model~\cite{dworkmodel} where there is a known bound 
$\Delta$ on message transmission delays, such that after an unknown time called \emph{GST} 
all transmissions arrive at their destinations within $\Delta$ bounds. 

{\bf System Guarantees}:
The goal is for replicas to form an agreement on a global ledger of transactions requested by clients and respond to clients with the outcome of executing transactions in sequential order.  
There are two requirements; {\em safety} is required under asynchrony and {\em liveness} 
is required under synchrony/\emph{GST}:

\begin{enumerate}[nosep]
    \item \textbf{Safety}: If two correct replicas $R$ and $R'$ commit two transactions $T$ and $T'$ at sequence number $k$ then $T = T'$.
    
    \item \textbf{Liveness}: Each correct replica will eventually commit a transaction $T$. 

\end{enumerate}

\section{Speculation in Streamlined Protocols}
\label{s:motivation}
Our primary goal is to reduce the latency for partially-synchronous streamlined consensus 
protocols. That is, we aim to bridge the gap between the latency of streamlined protocols and optimized 
stable-leader consensus protocols without losing a vital tenet: linearity. 
An additional goal of this work is to mitigate the slowness attacks and tail-forking attacks from streamlined protocols.

To this extent, we design \sysname{}, which uses two popular system design principles, speculation and slotting, to guarantee 
(1) low latency while maintaining linearity, (2) freedom from the slowness attack, and (3) resilience to tail-forking attack.

In the rest of this section, we discuss \HSOne{} (speculation) and defer discussion on slotting until \S\ref{s:slotting}.
To illustrate the challenges in introducing speculation to streamlined consensus protocols, 
we first briefly recap the skeleton of the \HSTwo~\cite{hs2} protocol.

\begin{figure*}[t]
    \centering
    \begin{tikzpicture}[scale=0.8, yscale=0.4]

        \draw[thick,draw=blue] (1.25,   4) edge ++(4.75, 0)
                               (7.25,   4) edge ++(4.75, 0)
                               (13.25,   4) edge ++(4.75, 0);
        
        \draw[thick,draw=black!75] (1.25,   0) edge ++(4.75, 0)
                                   (1.25,   1) edge ++(4.75, 0)
                                   (1.25,   2) edge ++(4.75, 0)
                                   (1.25,   3) edge ++(4.75, 0);
        
        \draw[thick,draw=black!75] (7.25,   0) edge ++(4.75, 0)
                                   (7.25,   1) edge ++(4.75, 0)
                                   (7.25,   2) edge ++(4.75, 0)
                                   (7.25,   3) edge ++(4.75, 0);

        \draw[thick,draw=black!75] (13.25,   0) edge ++(4.75, 0)
                                   (13.25,   1) edge ++(4.75, 0)
                                   (13.25,   2) edge ++(4.75, 0)
                                   (13.25,   3) edge ++(4.75, 0);

        \node[left,font=\tiny] at (1.3, 0) {$\Replica_3$};
        \node[left,font=\tiny] at (1.3, 1) {$\Replica_2$};
        \node[left,font=\tiny] at (1.3, 2) {$\Replica_1$};
        \node[left,font=\tiny] at (1.3, 3) {$\Replica_0$};
        \node[left,blue,font=\tiny] at (1.3, 4) {$C$};

        \node[left,font=\tiny] at (7.3, 0) {$\Replica_3$};
        \node[left,font=\tiny] at (7.3, 1) {$\Replica_2$};
        \node[left,font=\tiny] at (7.3, 2) {$\Replica_1$};
        \node[left,font=\tiny] at (7.3, 3) {$\Replica_0$};
        \node[left,blue,font=\tiny] at (7.3, 4) {$C$};

        \node[left,font=\tiny] at (13.3, 0) {$\Replica_3$};
        \node[left,font=\tiny] at (13.3, 1) {$\Replica_2$};
        \node[left,font=\tiny] at (13.3, 2) {$\Replica_1$};
        \node[left,font=\tiny] at (13.3, 3) {$\Replica_0$};
        \node[left,blue,font=\tiny] at (13.3, 4) {$C$};

        \node[left,font=\tiny] at (3.6, 4.8) {(i)};
        \node[left,font=\tiny] at (2.3, 3.5) {$\PP_0$};


        \node[left,font=\tiny] at (9.6, 4.8) {(ii)};
        \node[left,font=\tiny] at (8.3, 3.5) {$\PP_0$};
        \node[left,font=\tiny] at (11.3, 3.5) {$\PP_1$};


        \node[left,font=\tiny] at (15.6, 4.8) {(iii)};
        \node[left,font=\tiny] at (14.3, 3.5) {$\PP_0$};
        \node[left,font=\tiny] at (15.7, 3.5) {$\PP_1$};
        \node[left,font=\tiny] at (17.3, 3.5) {$\PP_2$};


        \draw[thin,draw=black!75] (1.5, 0) edge ++(0, 4)
                                  (3, 0) edge ++(0, 4)
                                  (4.5, 0) edge ++(0, 4);

        \draw[thin,draw=black!50] (2.25, 0) edge ++(0, 4)
                                  (3.75, 0) edge ++(0, 4)
                                  (5.25, 0) edge ++(0, 4)
                                  (3.75, 0) edge ++(0, -0.6)
                                  (5.25, 0) edge ++(0, -0.6);

        \draw[thin,draw=black!75] (7.5, 0) edge ++(0, 4)
                                  (9, 0) edge ++(0, 4)
                                  (10.5, 0) edge ++(0, 4);

        \draw[thin,draw=black!50] (8.25, 0) edge ++(0, 4)
                                  (9.75, 0) edge ++(0, 4)
                                  (11.25, 0) edge ++(0, 4)
                                  (9.75, 0) edge ++(0, -0.6)
                                  (11.25, 0) edge ++(0, -0.6);

        \draw[thin,draw=black!75] (13.5, 0) edge ++(0, 4)
                                  (15, 0) edge ++(0, 4)
                                  (16.5, 0) edge ++(0, 4);

        \draw[thin,draw=black!50] (14.25, 0) edge ++(0, 4)
                                  (15.75, 0) edge ++(0, 4)
                                  (17.25, 0) edge ++(0, 4)
                                  (15.75, 0) edge ++(0, -0.6)
                                  (17.25, 0) edge ++(0, -0.6);

        \path[->] (1.5, 3) edge ++(0.75, -1)
                         edge ++(0.75, -2)
                         edge ++(0.75, -3);

        \path[->] (3, 3) edge ++(0.75, -1)
                         edge ++(0.75, -2)
                         edge ++(0.75, -3);

        \path[->] (4.5, 3) edge ++(0.75, -1)
                         edge ++(0.75, -2)
                         edge ++(0.75, -3);
        


        \path[->] (7.5, 3) edge ++(0.75, -1)
                         edge ++(0.75, -2)
                         edge ++(0.75, -3);

        \path[->] (9, 3) edge ++(0.75, -1)
                         edge ++(0.75, -2)
                         edge ++(0.75, -3);

        \path[->] (10.5, 2) edge ++(0.75, 1)
                         edge ++(0.75, -1)
                         edge ++(0.75, -2);

        \path[->] (13.5, 3) edge ++(0.75, -1)
                         edge ++(0.75, -2)
                         edge ++(0.75, -3);

        \path[->] (15, 2) edge ++(0.75, 1)
                         edge ++(0.75, -1)
                         edge ++(0.75, -2);

        \path[->] (16.5, 1) edge ++(0.75, 1)
                         edge ++(0.75, 2)
                         edge ++(0.75, -1);

        \path[->] (2.25, 2) edge ++(0.75, 1)
                  (2.25, 1) edge ++(0.75, 2)
                  (2.25, 0) edge ++(0.75, 3);

        \path[->] (3.75, 2) edge ++(0.75, 1)
                  (3.75, 1) edge ++(0.75, 2)
                  (3.75, 0) edge ++(0.75, 3);

        \path[->] (5.25, 0) edge ++(0.75, 2)
                  (5.25, 1) edge ++(0.75, 1)
                  (5.25, 3) edge ++(0.75, -1);



        \path[->] (8.25, 2) edge ++(0.75, 1)
                  (8.25, 1) edge ++(0.75, 2)
                  (8.25, 0) edge ++(0.75, 3);

        \path[->] (9.75, 0) edge ++(0.75, 2)
                  (9.75, 1) edge ++(0.75, 1)
                  (9.75, 3) edge ++(0.75, -1);
                  
        \path[->] (11.25, 0) edge ++(0.75, 2)
                  (11.25, 1) edge ++(0.75, 1)
                  (11.25, 3) edge ++(0.75, -1);

        \path[->] (14.25, 3) edge ++(0.75, -1)
                  (14.25, 1) edge ++(0.75, 1)
                  (14.25, 0) edge ++(0.75, 2);

        \path[->] (15.75, 3) edge ++(0.75, -2)
                  (15.75, 2) edge ++(0.75, -1)
                  (15.75, 0) edge ++(0.75, 1);

        \path[->] (17.25, 3) edge ++(0.75, -3)
                  (17.25, 2) edge ++(0.75, -2)
                  (17.25, 1) edge ++(0.75, -1);


        \path[->,blue] (5.25, 3) edge ++(0.75, 1)
                  (5.25, 2) edge ++(0.75, 2)
                  (5.25, 1) edge ++(0.75, 3)
                  (5.25, 0) edge ++(0.75, 4);

        \path[->,blue] (9.75, 3) edge ++(0.75, 1)
                  (9.75, 2) edge ++(0.75, 2)
                  (9.75, 1) edge ++(0.75, 3)
                  (9.75, 0) edge ++(0.75, 4);

        \path[->,blue] (15.75, 3) edge ++(0.75, 1)
                  (15.75, 2) edge ++(0.75, 2)
                  (15.75, 1) edge ++(0.75, 3)
                  (15.75, 0) edge ++(0.75, 4);

        \path[->,blue] (17.25, 3) edge ++(0.75, 1)
                  (17.25, 2) edge ++(0.75, 2)
                  (17.25, 1) edge ++(0.75, 3)
                  (17.25, 0) edge ++(0.75, 4);

        \node[below,smalltext] at (1.85, 0) {\fontsize{5}{12}\selectfont\MName{Propose}};
        \node[below,smalltext] at (3.35, 0) {\fontsize{5}{12}\selectfont\MName{Prepare}};
        \node[below,smalltext] at (4.85, 0) {\fontsize{5}{12}\selectfont\MName{Commit}};
        \node[below,smalltext] at (2.6, 0) {\fontsize{5}{12}\selectfont\MName{Vote}};
        \node[below,smalltext] at (4.1, 0) {\fontsize{5}{12}\selectfont\MName{Vote2}};
        \node[below,smalltext] at (5.65, 0) {\fontsize{5}{12}\selectfont\MName{NewView}};

        \node[below,smalltext] at (7.85, 0) {\fontsize{5}{12}\selectfont\MName{Propose}};
        \node[below,smalltext] at (9.35, 0) {\fontsize{5}{12}\selectfont\MName{Prepare}};
        \node[below,smalltext] at (10.85, 0) {\fontsize{5}{12}\selectfont\MName{Propose}};
        \node[below,smalltext] at (8.6, 0) {\fontsize{5}{12}\selectfont\MName{Vote}};
        \node[below,smalltext] at (10.15, 0) {\fontsize{5}{12}\selectfont\MName{NewView}};
        \node[below,smalltext] at (11.6, 0) {\fontsize{5}{12}\selectfont\MName{Vote}};

        \node[below,smalltext] at (13.85, 0) {\fontsize{5}{12}\selectfont\MName{Propose}};
        \node[below,smalltext] at (15.35, 0) {\fontsize{5}{12}\selectfont\MName{Propose}};
        \node[below,smalltext] at (16.85, 0) {\fontsize{5}{12}\selectfont\MName{Propose}};
        \node[below,smalltext] at (14.65, 0) {\fontsize{5}{12}\selectfont\MName{NewView}};
        \node[below,smalltext] at (16.15, 0) {\fontsize{5}{12}\selectfont\MName{NewView}};

        \node[font=\tiny] at (3.75, -1) {Lock $\PP_0$};
        \node[font=\tiny] at (5.25, -1) {Execute $\PP_0$};

        \node[font=\tiny] at (9.75, -1) {Speculatively};
        \node[font=\tiny] at (9.75, -1.7) {Execute $\PP_0$};
        \node[font=\tiny] at (11.25, -1) {Commit $\PP_0$};

        \node[font=\tiny] at (15.75, -1) {Speculatively};
        \node[font=\tiny] at (15.75, -1.7) {Execute $\PP_0$};
        \node[font=\tiny] at (17.25, -1) {Commit $\PP_0$,};
        \node[font=\tiny] at (17.25, -1.7) {Speculatively};
        \node[font=\tiny] at (17.25, -2.4) {Execute $\PP_1$};
        
    \end{tikzpicture}     
\caption{Workflows of (i) \BHSTwo, (ii) \BHSOne, and  (iii) \SHSOne}
\label{fig:normal_case}
\end{figure*}


\begin{flushleft}
    \textbf{Recap of HotStuff-2}
\end{flushleft}

\HSTwo{} optimizes \HS{} 
by reducing commit latency by one phase (or two half-phases).
\HSTwo{} operates in a succession of {\em views} (Figure~\ref{fig:normal_case}(i)). In each view, 
a leader proposes a transaction $T$ and forms consecutive certificates on the initial proposal over two-and-half phases. 
In the first half-phase, the leader proposes the transaction $T$. In each subsequent phase:

\begin{enumerate}[nosep]
    \item Replicas generate threshold signature shares to ensure that at least $\n-\f$ replicas 
    accept the leader’s proposal and send it to the leader. 

    \item The leader aggregates threshold shares from $\n-\f$ replicas into a threshold signature, 
    which we refer to as a {\em certificate}, and broadcasts it to all the replicas.
\end{enumerate}

This chain of certificates guarantees safety as follows: The first certificate ({\em prepare-certificate}) 
guarantees non-equivocation by proving that it chains to a correct previous certificate and has the support 
of at least $\n-\f$ replicas. The second is a {\em commit-certificate}, a certificate-of-certificate, 
guaranteeing that $\n - \f$ replicas have received the prepare-certificate, and despite 
any $\f$ failures, $T$ will be committed. Replicas that learn the commit-certificate can mark 
$T$ {\em committed}, execute it, and return responses to the client; $T$ becomes committed to 
the immutable ledger.
These responses to the clients are often referred to as {\em finality confirmations}, as the corresponding transactions will never get revoked.
\\
\begin{flushleft}
\textbf{Sending \early{s}.}
\end{flushleft}

In the good case (no-failures), a \HSTwo{} client receives finality confirmations after two and half-phases (excluding the two network hops to receive client requests and send a response to the client).
With \HSOne{}, we want to cut down this delay to one and a half-phases.
\HSOne{} achieves this goal by making clients the first-class citizens of the consensus process--direct \textit{learners} of consensus decisions.
\HSOne{} requires replicas to employ speculative execution to serve clients with \early{s}.

Rather than requiring replicas to wait until they learn whether a transaction {\em has committed}, 
\HSOne{} allows replicas to speculate precisely when a transaction is prepared and highly likely to be committed by a quorum in \HSTwo{}. 
More specifically, 
replicas are allowed to speculate on a proposal in the second phase of the protocol, upon voting to commit a prepare-certificate.
%
Replicas execute a transaction $T$ as soon as they have the prepare-certificate for $T$ and send a response to the clients.
Thus, clients directly receive commit-votes and the result of executing $T$, which enables an \early.
When a client receives responses from a quorum of $\n-\f$ replicas, it learns two things: a transaction has been committed, and the execution result has been computed in advance. 
Safety follows from the commit-safety of \HSTwo{} because a client can determine if a commit-certificate will form.

In a non-speculative protocol, a client needs to collect only $\f+1$ execution responses to determine the finality because correct replicas execute a transaction once the commit decision is reached; response from just one correct replica guarantees commitment. 
However, in \HSOne{}, clients need to collect $\n-\f$ responses because $\f+1$ speculative responses only guarantee that one correct replica prepares the transaction. 
Commitment is guaranteed only when at least $\f+1$ correct replicas prepare the transaction.
Thus, a client learns that a transaction will finalize only upon collecting $\n-\f$ responses. 
Figure~\ref{fig:normal_case} (ii) depicts our \HSOne{} protocol.

Although clients of \HSOne{} wait for $\f$ additional messages (temporarily increasing memory footprint) compared to clients of \HS{} and \HSTwo{}, we argue that \HSOne{} clients do not incur higher latency because they receive early finality confirmations. 
In \S\ref{sec:eval}, we conduct several experiments to validate this claim. 
Moreover, a client can delay verifying and processing these additional responses as long as necessary while prioritizing other tasks. 
Such a delay would not impact the latency of \HSOne{} because replicas do not wait for any input from the client.

\begin{flushleft}
\textbf{The Prefix Speculation Dilemma.}
\end{flushleft}

In \HSOne{}, when a client receives a quorum of responses for a transaction $T$, it learns that $T$ will be committed and that finality has been reached for appending $T$ to the ledger in sequence order.
This decision also commits all transactions preceding $T$, and the result of executing $T$ reflects processing the {\em full prefix} of transactions up to and including $T$. 
However, responses for $T$ \textbf{must not} be combined with those for preceding transactions to form a quorum.
That is, say $T$ succeeds an earlier transaction $T'$ in sequence order, $T' \prec T$. 
The commit-votes (speculative responses) of $T$ must not be used as commit-votes of $T'$ in forming a commit-decision on $T'$.

This brings forth a challenging dilemma with respect to speculation \footnote{See Appendix~\ref{app:safe-spec1} for a detailed explanation of why the dilemma breaks {\em safety}.}: 
the responses from $T$ represent the execution of a full prefix ending with $T$.
When a replica $R$ speculatively executes $T$, it must execute all transactions that precede it. 
However, if $R$ did not commit the preceding transaction $T'$ prior to executing $T$, 
it {\bf must not} send responses for $T$ to clients because these responses represent commit-votes. Otherwise,
clients can mistakenly combine commit-votes from a partial quorum on $T'$ with commit-votes from another partial quorum on $T$
and assume that a decision has been reached on $T'$. 
On the other hand, $R$ must ensure there is \textbf{``no view gap''} between the view in which $T$ is prepared and its current view, to prevent speculative execution on a proposal that might be superseded by a higher certificate that is formed in the gap and unknown to $R$.

\begin{flushleft}
\textbf{Note on Speculation in Stable-Leader Protocols.}
\end{flushleft}
As stated in the introduction, the notion of applying speculative execution to \BFT{} protocols is not new~\cite{miguel-thesis,poe,zyzzyvaj,sbft}.
However, we argue that applying speculative execution to streamlined protocols is not a straightforward extension.

These stable-leader protocols stop speculative execution during the view change phases 
because they need to run an explicit view-change protocol.
At the end of the view-change protocol, all the replicas start the new view when they receive from the new leader a {\em state}.
This state starts from the last agreed-upon checkpoint, and for each sequence number that some replica claims to have 
observed since the last checkpoint, this state includes a prepare-certificate (if available) or a proposal from the previous leader.
However, before a replica can add any of these sequence numbers/proposals to their log, the leader needs to 
re-run consensus on each of them.

Even though stable-leader protocols require their clients to not combine votes on a transaction across views, 
the ability to stop consensus, change views, and re-run consensus on past transactions ensures that neither there is a situation
where a replica is executing a transaction $T$ but is yet to commit a preceding transaction $T'$ nor there is a ``view gap''.

\begin{flushleft}
\textbf{Tackling Prefix Speculation Dilemma.}
\end{flushleft}

Streamlined protocols {\bf do not have} the option of stopping the consensus and re-running consensus on past transactions. 
Thus, we state the following two rules to tackle the prefix speculation dilemma in streamlined protocols:

\begin{definition} \label{def:strongspec}
    \textbf{\strongspec{} Rule.} A replica $R$ can speculatively execute a transaction $T$ only if $T$ extends a prefix which is already known to commit.
\end{definition} 

\begin{definition} \label{def:nogap}
    \textbf{No-Gap Rule.} A replica $R$, currently in view $v$, may speculatively execute a transaction $T$ only if $T$ is proposed in view $v-1$ and a prepare-certificate of $T$ is formed in view $v$.
\end{definition}
\begin{flushleft}
    \textbf{Rollback.}
\end{flushleft}

Finally, we need to address the possibility that speculation does not succeed.
Upon speculatively executing $T$, a replica $R$ cannot commit $T$ to the (global) ledger yet as it does not know if $T$ will commit. 
Instead, each replica maintains a {\em local-ledger}, where it marks $T$ prepared and executed. 
If in a succeeding view, $R$ is about to speculatively execute a transaction $T'$ that conflicts with $T$ (See Definition ~\ref{def:conflicting}), then $R$ must perform a {\em rollback} operation in the local-ledger. 
$R$ can observe that at least $\n-\f$ replicas prepared $T'$ and then $T$ cannot commit.
Specifically, the replica should now
fetch the transaction $T'$ from other replicas, erase $T$ from its local-ledger, execute $T'$, 
add an entry for $T'$ to its global-ledger, and respond to the client.
We discuss this in more detail in \S\ref{ss:failures}.

\section{Speculative Core} 
\label{sec:design}
We first describe the variant of basic (non-streamlined) \HSOne{} variant; in \S\ref{ss:pipe-onephase}, we describe the streamlined \sysname{}.

\begin{figure}[t]

\begin{myprotocol}
            \TITLE{Local state}{replica $R$}
        \STATE $\Certificate{v_{lp}}$, $v_{lp}$: stores the highest known prepare certificate and its view number
        \STATE $\CCertificate{v_{lc}}$, $v_{lc}$: stores the highest known commit certificate and its view number
        \STATE $v$: current view
        \STATE $\Pending$: pending, uncommitted blocks of transactions
        \STATE local-ledger, global-ledger
        \SPACE

        \TITLE{Leader role}{running at leader $\Primary{v}$}

        \EVENT{Upon \MName{pacemaker.EnterView($v$)}}
            \STATE Wait until received $\n-\f$ \MName{NewView} messages for view $v$ \label{bhs1:recv-newview}
            \STATE Wait until received $\Certificate{v-1}$ \textbf{or} received $\n$ \MName{NewView} messages \textbf{or} \MName{pacemaker.ShareTimer}($v$)\label{bhs1:proposecondition}
            \STATE Let $B_v$ be a block of client transaction yet to be proposed \label{bhs1:upd-cert} 
            \STATE Broadcast $m = \SignMessage{\Message{Propose}{B_v, v, \Certificate{v_{lp}}, \CCertificate{v_{lc}}}}{\Primary{v}}$ \label{bhs1:send-propose}
        \ENDEVENT
        \SPACE

        \EVENT{Received $\n-\f$ \MName{NewView} messages with shares for $\Certificate{v-1}$} 
            \STATE $\CCertificate{v-1} \gets$ \MName{CreateThresholdSign}$(\n-\f ~\text{distinct} ~\Share{R}{C} ~\text{shares})$ \label{bhs1:create-commit-threshold}
        \ENDEVENT
        \SPACE

        \EVENT{Received $\n-\f$ \MName{ProposeVote} messages} \label{bhs1:recv-propvote}
            \STATE $\Certificate{v} \gets$ \MName{CreateThresholdSign}$(\n-\f ~\text{distinct}~\Share{R}{P}~\text{shares})$ \label{bhs1:create-threshold}
            \STATE Broadcast $\SignMessage{\Message{Prepare}{v,\Certificate{v}}}{\Primary{v}}$ \label{bhs1:send-prep}
        \ENDEVENT
        \SPACE

        \TITLE{Backup role}{running at each replica $R$ (including leader)}
        \EVENT{Received $\SignMessage{\Message{Propose}{B_v, \Certificate{w}, \CCertificate{x}}}{\Primary{v}}$} \label{bhs1:recv-propose}
            \STATE Execute all transactions up to (incl.) $B_x$, add result to global-ledger and respond to clients \COMMENT{traditional-commit rule}\label{bhs1:traditional-commit}
            \SPACE
            
            \IF{$w \ge v_{lp}$ \COMMENT{vote to prepare $B_v$}}
                \STATE $\Share{R}{P} \gets$ \MName{CreateThresholdShare}$(\Certificate{w}, v, Hash(B_v))$  \label{bhs1:create-threshold-share}
                \STATE Send $\SignMessage{\Message{ProposeVote}{v,\Share{R}{P}}}{R}$ to $\Primary{v}$ \label{bhs1:send-propvote}
            \ENDIF
        \ENDEVENT
        \SPACE

        \EVENT{Received $\SignMessage{\Message{Prepare}{v,\Certificate{v}}}{\Primary{v}}$} \label{bhs1:recv-prep}
                \IF{$\Certificate{v}$ extends $\Certificate{v-1}$ \COMMENT{prefix-commit rule}} \label{bhs1:prefix-commit-rule}
                    \STATE Execute all transactions up to (incl.) $B_{v-1}$, add result to global-ledger and respond to clients  \label{bhs1:spec-exec}
                \ENDIF
                \SPACE

                \IF{predecessor of $B_v$ is in global-ledger \COMMENT{\strongspec{} rule}}
                    \IF{local-ledger state conflicts with $B_v$} \label{bhs1:conflict-cert}
                        \STATE Roll local-ledger back to the common ancestor \label{bhs1:rollback}
                    \ENDIF
                    \STATE Execute all transactions in $B_v$ speculatively, add result to local-ledger and send client a response \COMMENT{speculatively execute $B_v$} \label{bhs1:spec-exec}
                \ENDIF
                \SPACE

                \STATE $\Share{R}{C} \gets$ \MName{CreateThresholdShare}$(\Certificate{v})$ \COMMENT{vote to commit $B_v$} \label{bhs1:commit-share}
                \STATE Send $\SignMessage{\Message{NewView}{v+1,\Certificate{v},\Share{R}{C}}}{R}$ to $\Primary{v+1}$   \label{bhs1:send-newview}   
                \STATE Call \MName{exitView}() \label{bhs1:call-exitview}
        \ENDEVENT
        \SPACE

        \EVENT{Upon timeout} \label{bhs1:timeout}
           \STATE Send $\SignMessage{\Message{NewView}{v+1,\Certificate{v_{lp}},\bot}}{R}$ to $\Primary{v+1}$.            
           \STATE Call \MName{exitView}() 
        \ENDEVENT
        \SPACE

        \FUNCTION{exitView}{}
            \STATE $v \gets v+1$. \label{bhs1:switch-view} \COMMENT{disable voting and speculative execution for view $v$}
            \STATE Call \MName{pacemaker.completedView}$()$\label{bhs1:pacemaker-completeview}
        \ENDFUNCTION

    \end{myprotocol}
    \caption{Basic \HSOne{}.}
    \label{alg:basic-hs1}
\end{figure}


\subsection{Non-Streamlined Speculation}\label{ss:basic-oneshot}
As we treat clients as first-class citizens, we start by describing the client's behavior.

{\bf Client Request.}
When a client $c$ wants the replicas to process its transaction $T$, it creates a 
\MName{Request} message including $T$ and sends it to one replica. 

{\bf Client Response.}
When a client $c$ receives identical \MName{Response} messages from $\n-\f$ replicas 
for its transaction $T$, it records this set of responses as an \early{} for $T$,
marks $T$ as {\em executed} and accepts the result of execution.

\paragraph{\bf Replica pseudocode.}
In Figure~\ref{alg:basic-hs1}, we present the pseudo-code for basic \HSOne{}.
Prior to describing the algorithm in detail, we lay down some useful definitions.

\begin{definition}
    {\em Prepare and Commit Certificates.} 
    A prepare-certificate $\Certificate{v}$ for a proposal $m$ aggregates $\n-\f$ threshold signature-shares for $m$ in view $v$. 
    A commit-certificate $\CCertificate{v}$ for a proposal $m$ aggregates $\n-\f$ threshold signature-shares for $\Certificate{v}$ in view $v$.
\end{definition}

\begin{definition}
    {\em Highest Known Certificate.}
    A certificate $\Certificate{v_{lp}}$ for view $v_{lp}$ is the highest prepare-certificate, known to 
    replica $R$.
    For brevity, we omit from the code explicitly updating $v_{lp}$ every time $R$ learns a new certificate.
\end{definition}

\begin{definition}
    {\em Extending Certificates.}
    Given two certificates $\Certificate{v}$ and $\Certificate{w}$, for views $v$ and $w$, 
    at a replica $R$, $\Certificate{v}$ extends $\Certificate{w}$ if $v > w$ 
    and $\Certificate{v}$'s construction includes $\Certificate{w}$.
    Further, if a certificate $\Certificate{k}$ extends $\Certificate{v}$ and 
    $\Certificate{v}$ extends $\Certificate{w}$, then transitively $\Certificate{k}$ 
    extends $\Certificate{w}$.
\end{definition}

\begin{definition}\label{def:conflicting}
    {\em Conflicting Certificates.} 
    Given two certificates, $\Certificate{v}$ and $\Certificate{w}$, for views $v$ and $w$, 
    at a replica $R$, $\Certificate{v}$ conflicts with $\Certificate{w}$ if neither 
    $\Certificate{v}$ extends $\Certificate{w}$, nor $\Certificate{w}$ extends $\Certificate{v}$.
\end{definition}

{\bf Local state} at a replica includes: 
(1) highest prepare-certificate, $\Certificate{v_{lp}}$, formed in view $v_{lp}$,
(2) highest commit-certificate, $\CCertificate{v_{lc}}$, formed in view $v_{lc}$,
(3) current view $v$,
(4) set of pending, uncommitted blocks of transactions $\Pending$, and
(5) the local-ledger and the global-ledger.

{\bf Propose.}
When the leader $\Primary{v}$ for view $v$--a replica $R$ with $v = \ID{R} \bmod \n$--enters 
view $v$, it waits to receive \MName{NewView} messages from at least $\n-\f$ replicas.
Each message carries the highest certificate known to its sender, which helps the leader learn the highest known certificate among them and update its $v_{lp}$.
Additionally, if these $\n-\f$ \MName{NewView} 
messages contain threshold signature-shares for $\Certificate{v-1}$, the leader forms a commit-certificate 
$\CCertificate{v-1}$ (Line~\ref{bhs1:create-commit-threshold}) and updates $\CCertificate{v_{lc}}$.
After learning the highest certificate across all correct replicas (Line~\ref{bhs1:proposecondition}, see further explanations in \S\ref{ss:failures}), the leader aggregates client transactions (yet to be proposed) into a 
block $B_v$ and creates a \MName{Propose} message $m$ that includes the view number $v$, $B_v$, $\Certificate{v_{lp}}$, and $\CCertificate{v_{lc}}$.
Then, $\Primary{v}$ broadcasts $m$ to all replicas (Lines~\ref{bhs1:upd-cert}-\ref{bhs1:send-propose}). 
{\em Note.} The \MName{Propose} message for view $0$, the genesis view, extends a hard-coded certificate that all replicas assume to be valid. 


{\bf ProposeVote.}
On receiving a \MName{Propose} message $m$ from $\Primary{v}$ (Line~\ref{bhs1:recv-propose}), 
a replica $R$ checks if the prepare certificate $\Certificate{w}$ in $m$ is not lower than its highest prepare-certificate $\Certificate{v_{lp}}$, i.e., $w \ge v_{lp}$.

If $w > v_{lp}$, then $R$ updates its $v_{lp}$ to $w$, sets $\Certificate{w}$ as the highest known prepare-certificate and fetches the block corresponding to $\Certificate{w}$ from other replicas. (\S\ref{ss:failures}).



If $w\ge v_{lp}$, $R$ creates a \MName{ProposeVote} message, which includes a threshold signature-share $\Share{R}{P}$ for $m$,
and sends this message to $\Primary{v}$ (Lines~\ref{bhs1:recv-propose}-\ref{bhs1:send-propvote}). Otherwise, $R$ ignores the message.

{\bf Prepare.}
When $\Primary{v}$ receives $\n-\f$ well-formed \MName{ProposeVote} messages for its 
proposal $m$, it combines their signature shares into a threshold signature to 
create a prepare-certificate $\Certificate{v}$ (Lines~\ref{bhs1:recv-propvote}-\ref{bhs1:create-threshold}). 
Then, $\Primary{v}$ creates a \MName{Prepare} message including $\Certificate{v}$ and broadcasts it (Line~\ref{bhs1:send-prep}).

{\bf Vote and Speculate on Prepare.}
On receiving a \MName{Prepare} message from the leader, a replica $R$ checks if the 
certificate $\Certificate{v}$ is a valid threshold signature for the leader's proposal $m$. If it is valid, $R$ updates its highest known prepare-certificate $\Certificate{v_{lp}}$ to $\Certificate{v}$.

If $B_v$'s predecessor is already in the global-ledger (i.e., meets the \strongspec{} rule) and 
$B_v$ was prepared in view $v$ (i.e., meets the No Gap rule\footnote{
We defined No Gap rule for streamlined protocols in \S~\ref{s:motivation}. However, it also implicitly applies to the non-streamlined versions: A replica $R$, currently in \MName{ProposeVote} phase of view $v$, may speculatively execute a transaction $T$ only if $T$ was proposed in the preceding \MName{Propose} phase of view $v$ and a prepare-certificate of $T$ is formed in view $v$.
}),
$R$ does the following (Lines~\ref{bhs1:recv-prep}-\ref{bhs1:spec-exec}):

\begin{enumerate}[nosep,wide]
    \item {\em Speculatively executes} the transactions in block $B_v$ of $m$.
    \item Send speculative responses with execution results to the respective clients. 
    \item Adds result of executing $B_v$ to its local-ledger.
\end{enumerate}

\Arxiv{{\em Note on execution model.}
Once the transactions are ordered, they are executed sequentially.
This paper focuses on reducing client latency caused by consensus.
Thus, we assume the simplest execution model: sequential execution of the ordered transactions.
Alternatively, other execution designs, such as parallel transaction execution, can be employed, but these require detecting and resolving conflicts among transactions.
}

{\bf ExitView and NewView.}
A replica $R$ exits view $v$ in two cases: upon receiving a prepare message from the leader and upon a timer expiration.
Prior to calling the \MName{exitView()} function, $R$ constructs a \MName{NewView} message, which includes $\Certificate{v_{lp}}$, and 
forwards it to the leader $\Primary{v+1}$ of view $v+1$. 
It then invokes the pacemaker to orchestrate view-synchronization as needed (Line~\ref{bhs1:pacemaker-completeview}). 

{\bf Commit.}
There are two commit rules in basic \HSOne{} ({\em traditional commit} and {\em prefix commit}), 
which dictate when a replica can write a block of transactions to the global-ledger. 

\begin{definition}
    {\em Traditional Commit Rule.}
    A replica marks a block $B_{v-1}$ as committed when it receives a commit-certificate $\CCertificate{v-1}$ 
    for $B_{v-1}$. 
\end{definition}

\begin{definition} \label{def:consecutive-commit-rule}
    {\em Prefix Commit Rule.}
    A replica marks a block $B_{v-1}$ as committed when it receives a prepare-certificate $\Certificate{v}$ that extends $\Certificate{v-1}$.
\end{definition}

As the name suggests, the traditional commit rule is common to any consensus protocol and has been 
used by all the protocols of the \HS{} family. 
Post speculatively executing the transaction, each replica creates a threshold share ($\Share{R}{C}$)
for the prepare-certificate and forwards this threshold share with the \MName{NewView} message to the leader of the next view (Lines~\ref{bhs1:commit-share}-\ref{bhs1:send-newview}).
Next, each replica calls the \MName{ExitView} procedure.
On receiving $\n-\f$ threshold shares for the same prepare-certificate, the leader of the next 
view combines them into a commit-certificate (Line~\ref{bhs1:create-commit-threshold}) and forwards it to all the replicas.
Upon receiving a commit-certificate $\CCertificate{v}$, a replica $R$ adds the block $B_v$ to the 
global-ledger and marks it committed (Line~\ref{bhs1:traditional-commit}).
{\em Note:} on receiving the commit-certificate, $R$ sends a response to a client if $R$ had not sent a speculative response for this transaction.


The prefix commit rule is an important optimization that allows correct replicas to commit 
blocks when \HSOne{} is experiencing replica failures, which we will expand on in \S~\ref{ss:failures}.

\subsection{Failures and Recovery Design}
\label{ss:failures}

A malicious replica can impact the consensus in various ways if it is the leader of an ongoing view:
(1) drop, delay, or prevent sending messages and/or certificates to prevent replicas from making progress, and
(2) equivocate by creating two proposals that extend the same certificate to prevent replicas from having the same state.
\HSOne{} should quickly detect these failures and resolve them to prevent performance degradation.

\begin{flushleft}
 {\bf Detecting lack of progress: Timeouts}   
\end{flushleft}

Like other protocols in the partial synchrony setting, \HSOne{} requires replicas to 
set timers. A replica $R$ starts a timer following the rules defined by the {\em pacemaker} protocol (\S\ref{ss:pacemaker}).
Upon timeout, a replica $R$ assumes that the leader of the current view (say $v$) has failed and 
thus sends a \MName{NewView} message to the leader of view $v+1$.
Post this, $R$ calls the \MName{ExitView} procedure  to move to the next view (Lines~\ref{bhs1:timeout}-\ref{bhs1:pacemaker-completeview}).

\begin{flushleft}
{\bf Lack of certificates from the last view}
\end{flushleft}

Leader $\Primary{v}$ of view $v$ may fail to receive the prepare-certificate $\Certificate{v-1}$ due to an unreliable network or faulty behaviors of the preceding leader. 
If it extends some lower certificate, its new proposal will get ignored by correct replicas that received $\Certificate{v-1}$. 
To ensure that the new proposal will be voted by all correct replicas, $\Primary{v}$ should wait for sufficiently long to receive the highest certificates known to all the correct replicas. Following the rules defined by the {\em pacemaker} protocol (\S\ref{ss:pacemaker}), after \emph{GST}, it is guaranteed that by \MName{pacemaker.ShareTimer}($v$), which is $3\Delta$ after $\Primary{v}$ enters view $v$, $\Primary{v}$ will receive \MName{NewView} messages including known certificates from all correct replicas. Thus, if $\Primary{v}$ did not receive $\Certificate{v-1}$, it should wait until either it received $\n$ \MName{NewView} messages or \MName{pacemaker.ShareTimer}($v$) (Line~\ref{bhs1:proposecondition}).

\begin{flushleft}
 {\bf Conflict Resolution: Rollback}   
\end{flushleft}
%
When a replica $R$ receives a prepare-certificate $\Certificate{v}$, 
\HSOne{} allows $R$ to set $\Certificate{v}$ as the highest known certificate and speculatively 
execute transactions of block $B_v$. A faulty leader may not send $\Certificate{v}$ to other replicas, 
in which case $B_v$ may not get committed.
To ensure replicas have a common state (global-ledger), \HSOne{} supports state rollback (or \textit{erasing local-ledger}). 

When a replica $R$ receives a prepare-certificate $\Certificate{w}$ in 
view $w$ for a proposal $m$, it speculatively executes $m$'s transactions and only 
updates its local-ledger; $R$ does not add $m$ to the global-ledger as it has only received a prepare-certificate for $m$ and has no guarantee that $m$ will commit in the future. 
Thus, $R$ can erase its local-ledger when it needs to roll back the effects of $m$'s transactions. Below is the condition for rollback:

\begin{definition}\label{def:rollback}
    {\em Rollback Condition.} Given two conflicting blocks $B_w$ and $B_v$ such that $w<v$, 
    if a replica $R$ speculatively executed transactions in $B_w$ with prepare-certificate $\Certificate{w}$, 
    $R$ will roll back $B_w$ when $R$ is about to speculatively execute the conflicting $B_v$ with prepare-certificate $\Certificate{v}$ (Lines~\ref{bhs1:conflict-cert}-\ref{bhs1:rollback}). 
\end{definition}

See Appendix~\ref{app:rollbackscenes} for a scenario illustrating rollback in \sysname{}.

\begin{flushleft}
 {\bf Prefix Commit: Processing Delayed Certificates}  
\end{flushleft}

Due to failures, replicas may vote on a proposal in a view but not receive a prepare-certificate for that proposal in the same view.
For example, the leader of view $v$ fails before broadcasting the prepare-certificate $\Certificate{v}$ for its proposal $m$ to at least $\n-\f$ replicas.
If such is the case, neither the client will receive an \early{} for $m$, nor the replicas will receive a commit-certificate for $m$ in view $v+1$.
So, how can we decide the fate of $m$?

If $m$ conflicts with another proposal $m'$ speculated in a view $w, w > v$, then it will be rolled back as described in \S~\ref{s:motivation}.
However, if there are no conflicts, that is, the leader of some view $x, x > v$ observes $\Certificate{v}$ and extends $\Certificate{v}$ in its proposal $m'$, 
a replica $R$ will execute transactions in $B_v$ and reply to the client once $R$ receives a commit-certificate $\CCertificate{x}$ for $m'$ (Line~\ref{bhs1:traditional-commit}).

Fortunately, we have an {\em optimization} that allows replicas to commit and execute $B_v$ at least one phase earlier; 
if $x = v+1$ and $\Certificate{v+1}$ is received, then a replica $R$ can commit $B_v$, execute transactions, add them to the global-ledger, and reply to their clients (Line~\ref{bhs1:prefix-commit-rule}), which we refer to as the prefix-commit rule.

\begin{flushleft}
    {\bf Recovery Mechanism}
\end{flushleft}

A faulty leader can skip broadcasting a
certificate to all the replicas. If any future leader has access to this valid certificate, 
it can extend its new proposal from this certificate. Such scenarios can occur in any protocol 
of the \HS{} family and are not limited to just malicious attacks; for example, 
a leader can crash before broadcasting the certificate to all replicas.


If the leader $\Primary{v}$ of view $v$ extends its proposal $m$ from the certificate $\Certificate{w}$, $w <v$, then each  
replica $R$ that receives the proposal needs to validate $\Certificate{w}$ and requires 
access to the corresponding proposal (say $m'$) of view $w$. If $R$ does not have access to $m'$, then it should fetch $R$ from other correct replicas, at least $\f+1$ of which should have it because they voted for $m'$.

\Arxiv{

\begin{figure}[t!]
    \begin{myprotocol}

    \FUNCTION{ShareTimer}{$v$}
    \RETURN $StartTime[v] + 3\Delta$
    \ENDFUNCTION
    \SPACE

    \FUNCTION{CompletedView}{}
    \IF{$v \mod \f+1 \neq 0$}
    \STATE Call \MName{EnterView($v$)}
    \ELSE
    \STATE Call \MName{SynchronizeEpoch($v$)}
    \ENDIF
    \ENDFUNCTION
    \SPACE

    \FUNCTION{SynchronizeEpoch}{$v$}
    \STATE $\Share{R}{} \gets$ \MName{CreateThresholdShare}$(v)$
    \STATE Send $\left<\MName{Wish}(v, \Share{R}{}) \right>_R$ to leaders $\Primary{v+k}, k = 0,1,2,...,\f$.
    \ENDFUNCTION

    \SPACE
    \TITLE{Epoch Leader role}{running at leader $\Primary{v+k}$, $k = 0,1,2,...,\f$.}
        \EVENT{Upon receiving $\n-\f$ \MName{Wish} messages of view $v$}
            \STATE $TC_v \gets$ \MName{CreateThresholdSignature}$(\n-\f ~\text{distinct} ~\Share{r} ~\text{ shares})$
            \STATE Broadcast $TC_v$.
        \ENDEVENT

    \SPACE
    \TITLE{Epoch Backup role}{running at each replica $R$}
        \EVENT{Upon receiving $TC_v$ at time $t$}
            \STATE Relay $TC_v$ to the leaders $\Primary{v+k}, k = 0,1,2,...,\f$
            \FOR{$k \gets 0, 1, 2, ..., \f$}
            \STATE $StartTime[v+k] \gets t+k\tau$ 
            \ENDFOR
            \STATE Call \MName{EnterView($v$)}
        \ENDEVENT

    \end{myprotocol}
    \caption{Pseudocode of Pacemaker Protocol}\label{alg:pacemaker}
\end{figure}
}

\subsubsection{\bf Pacemaker}
\label{ss:pacemaker}
For a system to make {\em progress}, at least $\n-\f$ correct replicas should be in the same view. 
Otherwise, a leader cannot collect enough votes to make progress and to generate a prepare-certificate (\S\ref{ss:pipe-onephase}). 
Specifically, under an unreliable network or when the leader is faulty, correct replicas can diverge: 
some replicas may have progressed to higher views, while others are stuck on an old view.
To prevent this divergence among correct replicas, we adopt the \emph{pacemaker} designs of prior works~\cite{raresync,lewis};
group views into {\em epochs}, each of which contains $\f+1$ consecutive views, and conduct view synchronization at the beginning of every epoch. 

\Arxiv{
In Figure~\ref{alg:pacemaker}, we illustrate the pseudocode for pacemaker. 
Every time a replica $R$ reaches at the end of a view, it calls the function \MName{CompletedView} (Lines 3-7) to check if the next view (say $v$) is part of the current epoch.
If this is the case, $R$ enters view $v$.
Otherwise, $v$ is the first view of the next epoch ($v\!\mod\!(\f+1) = 0$) and $R$ must synchronize its view with the other replicas. 
$R$ calls the function \MName{SynchronizeView}($v$) (Lines 8-10) and delays entering the view $v$ until the view synchronization is complete. 

The function \MName{SynchronizeView}($v$) requires $R$ to send a \MName{Wish} message to the $\f+1$ leaders of the next epoch; $\Primary{v+k}$, where $k = 0,1,2,...,\f$.
When a  leader of the next epoch receives $\n-\f$ \MName{Wish} messages for view $v$, it creates a \emph{Timeout Certificate} $TC_v$ and broadcasts it to all the replicas (Lines 14-15). 
Any non-leader replica $R$ that receives $TC_v$ forwards this certificate to all the $\f+1$ leaders for the next epoch. 
Next, $R$ sets the {\em starting time} for each of the next $\f+1$ views $v+k$, $k=0,1,2,...,\f$.
Say, $R$ received $TC_v$ at time $t$, then view $v+k$ starts at time $t+k\tau$, where $\tau$ is a predetermined timer length that is sufficiently long for a non-faulty leader to reach a consensus on the proposal of its view.
{\em Note:} the starting time for view $v+k$ is also the timeout for view $v+k-1$.
Post this, $R$ enters the next view $v$ (Lines 16-18).
} 

The \emph{pacemaker} guarantees that, after \emph{GST}, once the first synchronization is done at view $v_s$, if a correct replica enters view $v, v\ge v_s$ at time $t$ and sets its timer for view $v$ to expire at time $t'$, then all correct replicas will enter view $v$ before $t+2\Delta$ and no correct replica will time out and enter view $v+1$ before $t'-2\Delta$, where $\Delta$ is the transmission delay bound.
Post $t+2\Delta$, if the leader $L_{v}$ for view $v$ waits for an additional message delay, $\Delta$ then it is guaranteed to receive \MName{NewView} messages from all the correct replicas and learn the highest known certificate. 
Thus, the function \MName{ShareTimer}($v$) returns after $t+3\Delta$.

\section{Streamlined Speculation}
\label{ss:pipe-onephase}

Basic \HSOne{} (\S\ref{ss:basic-oneshot}) processes only one proposal every two phases. Like \HS{}, we can {\em streamline the phases} of \HSOne{} to ensure that 
we rotate leaders and inject a new proposal every phase.
This has the potential to increase throughput by $2 \times$. 

Borrowing from the streamlined variant of \HS{}, streamlined \HSOne{} works as follows:
it overlaps the second phase of view $v$, consisting of \emph{Prepare} and \emph{NewView} 
steps, with the first phase of view $v+1$, namely, \emph{Propose} and \emph{ProposeVote} 
steps. Each view (or leader) lasts for only one phase. The leader of each view waits 
for $\n-\f$ \MName{NewView} messages from the preceding view. The leader first attempts 
to create a prepare-certificate from the threshold shares it received from the replicas.  
It then selects the highest prepare-certificate it knows and references it in a new 
proposal with a new batch of client transactions. 

{\bf Commit Rule.}
Unlike the basic \HSOne{}, the streamlined design has only one commit rule: 
replicas follow the prefix commit rule (Definition \ref{def:consecutive-commit-rule}) to add a transaction to the global-ledger. 
As each view consists of one phase,
there is no explicit opportunity to create a commit-certificate.
In view $v$, a replica $R$ commits a block $B_{w-1}$, proposed in view $w-1$, if the proposal 
of view $v$ includes the certificate $\Certificate{w}$ that extends the certificate $\Certificate{w-1}$.
{\em Note:} We no longer distinguish between prepare and commit certificates as in basic \HSOne{}.

{\bf \strongspec{} Rule and No-Gap Rule.} 
As in the basic variant, 
rules guaranteeing safe speculation are 
needed in streamlined \HSOne{} to tackle the Prefix Speculation dilemma described in \S\ref{s:motivation}. The enforcement of the \strongspec{} rule is similar to the basic 
regime: \textit{a replica $R$ can speculate on a block $B_v$ provided that the prefix of $\Certificate{B_v}$ is committed.} 
See Appendix~\ref{app:safe-spec2} for examples of not following the \strongspec{} rule and the No Gap in streamlined \HSOne{}. Similarly, enforcement of the No-Gap rule (Definition~\ref{def:nogap}) is necessary, that is, $w = v-1$.

\begin{figure}[t]
\begin{myprotocol}
        \TITLE{Leader role}{running at leader $\Primary{v}$}
        \EVENT{Upon \MName{pacemaker.EnterView()}}
            \STATE Wait until received $\n-\f$ \MName{NewView} messages  for view $v$ \label{shs1:recv-newview}
            \STATE Wait until $\Primary{v}$ forms a certificate $\Certificate{v-1}$ \textbf{or}  received $\n$ \MName{NewView} messages \textbf{or} \MName{pacemaker.ShareTimer}($v$) \label{shs1:wait} 
            \STATE Let $B_v$ be a block of client transaction yet to be proposed 
            \STATE Broadcast $m = \SignMessage{\Message{Propose}{B_v, v, \Certificate{v_{lp}}}}{\Primary{v}}$ \label{shs1:send-propose}
        \ENDEVENT
        \SPACE

        \EVENT{Received $\n-\f$ \MName{NewView} messages with shares for the same proposal of view $v-1$} 
            \STATE $\Certificate{v-1} \gets$ \MName{CreateThresholdSign}$(\n-\f ~\text{distinct}~\Share{R}{}~\text{shares})$
            \label{shs1:create-threshold} 
        \ENDEVENT
        \SPACE

        \TITLE{Backup role}{running at each replica $R$ (including leader)}
        \EVENT{Received $\SignMessage{\Message{Propose}{B_v, \Certificate{w}}}{\Primary{v}}$} \label{shs1:recv-propose}
            \IF{$\Certificate{w}$ extends $\Certificate{w-1}$ \COMMENT{commit-rule}} \label{shs1:commit-rule}
                \STATE Execute all transactions up to (incl.) $B_{w-1}$, add result to global-ledger and respond to clients  \label{shs1:exec}
            \ENDIF
            \SPACE

            \IF{$w = v-1$\COMMENT{No-Gap rule}}\label{shs1:nogap}
            \IF{predecessor of $\Certificate{v-1}$ is in global-ledger\COMMENT{\strongspec{} rule}}\label{shs1:prefix-spec}
                \IF{local-ledger state conflicts with $B_{v-1}$} \label{shs1:conflict-cert}
                    \STATE Rollback local-ledger to the common ancestor \label{shs1:rollback}
                \ENDIF
                \STATE Execute all transactions in $B_{v-1}$ speculatively, add result to local-ledger and send client a response \COMMENT{speculatively execute $B_{v-1}$} \label{shs1:spec-exec}
            \ENDIF
            \ENDIF
            \SPACE

            \IF{$w \ge v_{lp}$}\label{shs1:highest-cert}
                \STATE $\Share{R}{} \gets$ \MName{CreateThresholdShare}$(\Certificate{w},v,Hash(B_{v}))$  \label{shs1:create-threshold-share}
                \STATE Send $\SignMessage{\Message{NewView}{v+1,\Certificate{w},\Share{R}{}}}{R}$ to $\Primary{v+1}$ \label{shs1:send-newview}
            \ENDIF
            \STATE Call \MName{exitView}() \label{shs1:call-exit}
        \ENDEVENT
        \SPACE

        \EVENT{Upon timeout} \label{shs1:timeout}
            \STATE Send $\SignMessage{\Message{NewView}{v+1,\Certificate{v_{lp}},\bot}}{R}$ to $\Primary{v+1}$. \label{shs1:timeout-newview}
            \STATE Call \MName{exitView}() \label{shs1:call-exit-view}
        \ENDEVENT
        \SPACE

        \FUNCTION{exitView}{}
            \STATE $v \gets v+1$. \label{shs1:switch-view} \COMMENT{disable voting for view $v$}
            \STATE Call \MName{pacemaker.completedView}$()$\label{shs1:pacemaker-completeview}
        \ENDFUNCTION

    \end{myprotocol}
    \caption{Streamlined \HSOne{}.}
    \label{alg:steamlined-hs1}
\end{figure}


\subsection{Streamlined HotStuff-1 Protocol} 
The streamlined protocol is reduced into a single phase of (1) \textit{propose} and 
(2) \textit{vote} that includes the speculative execution as demonstrated in 
Figure~\ref{fig:normal_case} (iii).

{\bf Propose.}
When the leader $\Primary{v}$ for view $v$ receives 
well-formed \MName{NewView} messages from at least $\n-\f$ replicas (Figure~\ref{alg:steamlined-hs1} Line~\ref{shs1:recv-newview}), it tries
to combine their threshold signature-shares into a threshold signature to create 
a certificate $\Certificate{v-1}$ for view $v-1$ (Line~\ref{shs1:create-threshold}).
If $\Primary{v}$ fails, it keeps waiting for more \MName{NewView} messages until it forms a certificate $\Certificate{v-1}$ or it receives $\n$ \MName{NewView} messages or \MName{pacemaker.ShareTimer($v$)} (Line~\ref{shs1:wait}).
Then, the leader extends its highest certificate to form its new proposal 
as a $\MName{Propose}$ message $m$ and broadcasts it to all replicas. This proposal includes the view number $v$, 
a block $B_v$ of client transactions yet to be proposed, and $\Certificate{v_{lp}}$ (Line~\ref{shs1:send-propose}).

{\bf Execute and Ledger Update.}
On receiving a \MName{Propose} message (let's call it $m$) from the leader (Line~\ref{shs1:recv-propose}), $R$ does the following (Lines~\ref{shs1:commit-rule}-\ref{shs1:call-exit}):

\begin{enumerate}[nosep,wide]
    \item Following the commit-rule: if $\Certificate{w}$ extends $\Certificate{w-1}$, then $R$ executes transactions 
    for all blocks up to $B_{w-1}$ (blocks that $B_{w-1}$ extends) if yet to be executed, adds them to the global-ledger and sends a reply to 
    respective clients (Lines~\ref{shs1:commit-rule}-\ref{shs1:exec}).
    
    \item If $w = v-1$ (meets the No-Gap rule), then following the \strongspec{} Rule: if the predecessor of $B_{v-1}$ is committed, then $R$ {\em speculatively executes} the transactions in blocks $B_{v-1}$, adds them to the local-ledger, and sends a reply to respective clients. Before speculatively executing $B_{v-1}$, $R$ first rolls back its local-ledger if it has speculatively executed a conflicting block (Lines~\ref{shs1:nogap}-\ref{shs1:spec-exec}). 

    \item Finally, $R$ checks if $w$, the view of the certificate $\Certificate{w}$ in $m$, is not lower than its $v_{lp}$. If $w \ge v_{lp}$, $R$ updates its highest known certificate $\Certificate{v_{lp}}$ with $\Certificate{w}$. Then, $R$ creates a \MName{NewView} message including $\Certificate{v_{lp}}$ and a threshold signature-share $\Share{R}{}$ for $m$, and sends it to the leader of the next view, $\Primary{v+1}$ (Lines~\ref{shs1:highest-cert}-\ref{shs1:send-newview}).
\end{enumerate}

{\bf Timer expiration.}
In case of timer expiration, the replica $R$ constructs a \MName{NewView} message, which includes an empty threshold signature-share and the highest known certificate $\Certificate{v_{lp}}$, and 
forwards it to the leader $\Primary{}$ for view $v+1$ (Lines~\ref{shs1:timeout}-\ref{shs1:call-exit-view}). 

{\bf ExitView and NewView.}
Like earlier, a replica $R$ is ready to exit view $v$ in two cases: upon receiving a \MName{Propose} message from the leader and upon a timer expiration.
\textsc{ExitView}$()$ invokes the pacemaker to orchestrate view-synchronization as needed (Line~\ref{shs1:pacemaker-completeview}).

{\bf Correctness Proof.}
See Appendix~\ref{app:proof} for the correctness proof.
\section{Slotting}
\label{s:slotting}

Rotating leaders in \BFT{} protocols leads to the following challenges:
\begin{enumerate}[wide,nosep]
    \item {\bf Leader-slowness phenomenon.}
    Rational leaders, who are not malicious but aim to maximize their gains, may delay proposing a block of transactions until as late as possible in their rotation, as they are incentivized to include transactions that yield higher fees.
    Similarly, block builders participating in a proposer-builder auction may also delay to maximize MEV (maximal extractable value) exploits~\cite{Daian2019FlashB2,pbs2024,time-role-mev}. If a leader/builder proposes its block too early, it risks filling the block with transactions that offer lower fees than those that may come in the future.
    Thus, rational leaders and builders may slow down progress, causing increased client latency.

\Arxiv{
\begin{example}\label{ex:slowleader}
Assuming that each block can include at most $100$ transactions and the maximum allowed time for a view to complete is $4s$, while it takes a leader approximately $1s$ to create a block and ensure that its proposal completes all phases of \HSOne{}.
In an ideal case, the latency for each transaction would be $\approx 1s$.
A rational leader will wait for four seconds to create the block in the hope of selecting the top $100$ highest fees paying transactions, 
which ensures the average latency to be $\approx 4s$.
\end{example}
}
    \item {\bf Tail-forking attack.} In streamlined protocols, the two protocol phases necessary 
    to commit a transaction are spread across the reign of two leaders. The second leader, if malicious, 
    may skip the proposal from the previous leader by pretending that it did not receive enough votes for it, 
    instead of helping drive it to a commit decision.  

\Arxiv{
\begin{example}\label{ex:slotting}
Assuming that $R_0$ and $R_1$ are the leaders for views $v$ and $v+1$, respectively, and $R_1$ is malicious.
In view $v$, $R_0$ broadcasts a \MName{Propose} message for $B_v$ containing $\Certificate{v-1}$, and all replicas send a \MName{ProposeVote} message for $B_v$ to $R_1$.
As $R_1$ is malicious, in view $v+1$, assume that $R_1$ initiates
the tail-forking attack by ignoring the \MName{NewView} messages for $B_v$ and broadcasts a \MName{Propose} message for $B_{v+1}$ 
that includes the certificate $\Certificate{v-1}$.
Since no replica has access to a higher known certificate than $\Certificate{v-1}$, all replicas accept $B_{v+1}$, create a threshold signature-share for $B_{v+1}$, and send it with a \MName{NewView} message to $R_2$. 
Consequently, all the work done during view $v$ is a waste.
\end{example}
}
    
\end{enumerate}

We address these challenges by introducing \textit{slotting} into the core of streamlined consensus protocols. Slotting enables each leader to propose multiple blocks—one per slot—within its view rotation period. An adaptive slotting mechanism allows leaders to propose as many slots as possible before the view timer expires.
Assigning multiple slots per leader/view offers two key benefits:
(1) It incentivizes timely proposal of available transactions, as proposing more blocks yields greater rewards; and
(2) It eliminates tail-forking attacks for all but the final slot, since a leader can extend its own slot to prevent forking.
\Add{Additionally, we introduce \emph{carry blocks} in first-slot proposals to protect the last slot of the previous view, provided that at least $\f+1$ correct replicas have voted for it.}

\Arxiv{
With slotting, assuming Example~\ref{ex:slowleader}, we expect each leader to propose at least four blocks (one per slot) per view with latency $\approx 1s$; assuming Example~\ref{ex:slotting},
$R_1$ can only tail-fork the last slot of view $v$, but three out of four blocks will reach consensus.
}


\subsection{Slotting Design}
\label{ssec:slotting}
We proceed to describe how to incorporate a slotting design into streamlined \HSOne{}. 
{\em Note:} Our design of slotting is applicable to any protocol of the \HS{} family. 

We introduce two additional notations:

{\em First}, we enumerate leader proposals with a pair of numbers: a leader/view number and a slot 
number within the view. Blocks are ordered lexicographically: if $v < v'$, then block $B_{i, v}$ 
is ordered lower than $B_{i', v'}$. If $v = v'$ and $i < i'$, then block $B_{i, v}$ is ordered 
lower than $B_{i', v'}$. 
For instance, in Figure ~\ref{fig:snake}, we illustrate a chain of blocks generated under the slotting 
design. Each block extends a certificate of the preceding one, resulting in a \textit{snake-like} chain that threads blocks within each view and, at the end of each view, threads to the next view. 
In the figure, block $B_{2,1}$ includes a certificate for $B_{1,1}$, block $B_{1,2}$ includes a certificate for $B_{4,1}$, and so on.
%

{\em Second}, we introduce a new message type, \MName{NewSlot}, to distinguish a replica’s transition to a new slot within the same view from its transition to a new view. Both \MName{NewSlot} and \MName{NewView} messages contain threshold signature shares that serve as votes, enabling consensus over slot and view transitions, respectively.
To differentiate these votes, replicas sign not only the proposal but also distinct contextual parameters—\emph{New-Slot} and \emph{New-View}. As a result, we define two types of certificates: \emph{New-Slot} and \emph{New-View}. Each \emph{New-View} certificate is further annotated with a parameter ${fv}$, indicating the view in which it was formed, i.e., it is formed by $\Primary{fv}$.

Next, we describe the protocol modifications needed to support slotting. 
%
As before, a replica maintains pending, uncommitted blocks of transactions, a local-ledger and the global-ledger.
The local state at a replica (refer to Figure~\ref{alg:slot-steamlined-hs1-primary}) includes: 
(1) $\Certificate{s_{lp}, v_{lp}}$, the highest known certificate of view $v_{lp}$, slot $s_{lp}$,
(2) $s, v$, the current slot and view,
(3) $B_h$, the highest voted block with hash $H_h$.
%


\Add{
A well-formed first-slot proposal in view $v$ must provide a \emph{self-contained proof} of \emph{``no tail-forking''} in one of two ways:
(i) \textbf{form} and \textbf{extend} a \emph{New-View} certificate using votes in $\n-\f$ \MName{NewView} messages sent to $\Primary{v}$, e.g., in Figure~\ref{fig:snake}, $B_{1,2}$ extends $\Certificate{4,1}$; or
(ii) \textbf{extend} its highest certificate $\Certificate{s_{lp}, v_{lp}}$ and \textbf{carry} a block $B_u$. 

\begin{definition}
\label{def:carryblock} 
\textbf{Carry Block}: The \emph{lowest uncertified block} $B_u$ that extends the certificate $\Certificate{s_{lp}, v_{lp}}$.
\begin{itemize} \item If $\Certificate{s_{lp}, v_{lp}}$ is a \emph{New-View} certificate formed in view $fv$, then $B_u$ is $B_{1,fv}$. For example, $B_{1,3}$ extends the \emph{New-View} certificate $\Certificate{4,1}$ with $fv=2$ and carries the uncertified block $B_{1,2}$. 
\item If $\Certificate{s_{lp}, v_{lp}}$ is a \emph{New-Slot} certificate, then $B_u$ is $B_{s_{lp}+1,v_{lp}}$. For instance, $B_{1,1}$ extends the \emph{New-Slot} certificate $\Certificate{3,0}$ and carries the uncertified block $B_{4,0}$. \end{itemize} \end{definition}

The notions of the \emph{self-contained proof of ``no tail-forking''} and  \emph{carry block} guarantee that if a correct leader $\Primary{v}$ proposed at least two slots in view $v$ and $\f+1$ correct replicas have voted for its last slot, then view-$v$ slots are protected from \emph{tail-forking attacks}. See further explanations in \S~\ref{sec:tailfork}.

}

Figure ~\ref{alg:slot-steamlined-hs1-primary} and ~\ref{alg:slot-steamlined-hs1-replica} illustrate the pseudocode of streamlined \HSOne{} with slotting.
\begin{figure}[t]
    \begin{tikzpicture}[yscale=0.6]

        \node[align=center,anchor=north,smalltext] at (0.4,4) { $View 0$};
        \node[align=center,anchor=north,smalltext] at (2.4,4) { $View 1$};
        \node[align=center,anchor=north,smalltext] at (4.4,4) { $View 2$};
        \node[align=center,anchor=north,smalltext] at (6.4,4) { $View 3$};

        \node[align=center,anchor=north,smalltext] at (-1,3.4) { $Slot_1$};
        \node[align=center,anchor=north,smalltext] at (-1,2.4) { $Slot_2$};
        \node[align=center,anchor=north,smalltext] at (-1,1.4) { $Slot_3$};
        \node[align=center,anchor=north,smalltext] at (-1,0.4) { $Slot_4$};

        \draw[black, thick] (0,3) rectangle (0.8,3.5);
        \draw[black, thick] (2,3) rectangle (2.8,3.5);
        \draw[orange, thick, dashed] (4,3) rectangle (4.8,3.5);
        \draw[black, thick] (6,3) rectangle (6.8,3.5);
        \draw[black, thick] (0,1) rectangle (0.8,1.5);
        \draw[black, thick] (2,1) rectangle (2.8,1.5);
        \draw[black, thick] (6,1) rectangle (6.8,1.5);
        \draw[orange, thick, dashed] (0,0) rectangle (0.8,0.5);
        \draw[orange, thick] (2,0) rectangle (2.8,0.5);
        \draw[black, thick] (6,0) rectangle (6.8,0.5);
        \draw[black, thick] (0,2) rectangle (0.8,2.5);
        \draw[black, thick] (2,2) rectangle (2.8,2.5);
        \draw[black, thick] (6,2) rectangle (6.8,2.5);

        \node[align=center,anchor=north,smalltext] at (0.4,0.5) {$B_{4,0}$};
        \node[align=center,anchor=north,smalltext] at (0.4,1.5) {$B_{3,0}$};
        \node[align=center,anchor=north,smalltext] at (0.4,2.5) {$B_{2,0}$};
        \node[align=center,anchor=north,smalltext] at (0.4,3.5) {$B_{1,0}$};
        \node[align=center,anchor=north,smalltext] at (2.4,0.5) {$B_{4,1}$};
        \node[align=center,anchor=north,smalltext] at (2.4,1.5) {$B_{3,1}$};
        \node[align=center,anchor=north,smalltext] at (2.4,2.5) {$B_{2,1}$};
        \node[align=center,anchor=north,smalltext] at (2.4,3.5) {$B_{1,1}$};
        \node[align=center,anchor=north,smalltext] at (4.4,3.5) {$B_{1,2}$};
        \node[align=center,anchor=north,smalltext] at (6.4,0.5) {$B_{4,3}$};
        \node[align=center,anchor=north,smalltext] at (6.4,1.5) {$B_{3,3}$};
        \node[align=center,anchor=north,smalltext] at (6.4,2.5) {$B_{2,3}$};
        \node[align=center,anchor=north,smalltext] at (6.4,3.5) {$B_{1,3}$};

        \draw[->, blue] (0.4,0.5) -- (0.4,1);
        \draw[->, blue] (0.4,1.5) -- (0.4,2);
        \draw[->, blue] (0.4,2.5) -- (0.4,3);
        \draw[->, blue] (2.4,0.5) -- (2.4,1);
        \draw[->, blue] (2.4,1.5) -- (2.4,2);
        \draw[->, blue] (2.4,2.5) -- (2.4,3);
        \draw[->, blue] (6.4,0.5) -- (6.4,1);
        \draw[->, blue] (6.4,1.5) -- (6.4,2);
        \draw[->, blue] (6.4,2.5) -- (6.4,3);

        \draw[->] (1.4,0.3) -- (0.8,0.3);
        \draw[->] (3.4,0.3) -- (2.8,0.3);
        \draw[->] (5.4,3.3) -- (4.8,3.3);

        \draw (2,3.3) -- (1.4,3.3);
        \draw (4,3.3) -- (3.4,3.3);
        \draw (6,3.3) -- (5.4,3.3);

        \draw (1.4,3.3) -- (1.4,0.3);
        \draw (3.4,3.3) -- (3.4,0.3);

    \end{tikzpicture}
    \caption{Chain in HotStuff-1 with Slotting, in which solid black blocks are extended with \emph{New-Slot} certificates. To provide a \emph{self-contained proof} of \emph{``no tail-forking''} in first-slot proposals, solid orange blocks are \emph{extended} in way (i), with a \emph{New-View} certificate formed by the next leader; while shaded orange blocks are \emph{carried} in way (ii), without a certificate.}
    \label{fig:snake}
\end{figure}

{\bf Propose.}
At each slot $s$, the leader $\Primary{v}$ for view $v$ awaits messages from at least $\n{-}\f$ 
replicas of either of the following types: 

\begin{enumerate}[nosep]
    \item well-formed \MName{NewView} messages for view $v$, if $s = 1$,  or 
    \item well-formed \MName{NewSlot} messages for slot $(s-1, v)$ if $s > 1$.
\end{enumerate}
Thus, $\Primary{v}$ administers two types of transitions.

{\em NewView:} The first is entering a new view. The leader awaits well-formed \MName{NewView} 
messages for view $v$ from at least $\n{-}\f$ replicas. 
%
$\Primary{v}$ delays proposing its first-slot block $B_{1,v}$ until any of the following conditions is met:

\begin{enumerate}
    \item \label{condition1} A \emph{New-View} certificate $\Certificate{s_{w}, w}, w < v$, can be formed with $\n{-}\f$ \MName{NewView} messages containing \emph{New-View} threshold signature-shares for the same proposal of slot $(s_{w}, w)$.
    \item \label{condition2} $\Primary{v}$ received $\n$ \MName{NewView} messages.
    \item \label{condition3} \MName{pacemaker.ShareTime($v$)}.
    \item \label{condition4} For any slot higher than $\Certificate{s_{lp}, v_{lp}}$, $\Primary{v}$ received $\n{-}\f$ \MName{New-View} messages that do not vote for it.
\end{enumerate}

With these four conditions, it is guaranteed that after \emph{GST}, a correct $\Primary{v}$ can learn the highest certificate across all correct replicas, either by forming it by itself through (1) or learning it from others through (2)-(4). See further explanations in \S~\ref{sec:slot_responsiveness}.

If condition (1) is satisfied, $\Primary{v}$ proposes its first-slot proposal in way (i): it forms a \emph{New-View} certificate $\Certificate{s_w, w}$ and updates its local highest certificate $\Certificate{s_{lp}, v_{lp}}$ to $\Certificate{s_w, w}$. Then, $\Primary{v}$ broadcasts a \MName{Propose} message $m$ containing the block $B_{1,v}$, a batch of new transactions, and the updated certificate $\Certificate{s_{lp}, v_{lp}}$. For example, in Figure~\ref{fig:snake}, $\Primary{1}$ proposes $B_{1,1}$ extending the \emph{New-View} certificate $\Certificate{4,0}$ (see Figure~\ref{alg:slot-steamlined-hs1-primary}, Lines~\ref{slotleader:condition1}–\ref{slotleader:endcondition1}).

If (2)-(4) is satisfied, $\Primary{v}$  proposes its first-slot proposal in way (ii): $\Primary{v}$ broadcasts a \MName{Propose} message $m$ that contains $B_{1,v}$, $\Certificate{s_{lp}, v_{lp}}$, and $H_{u}$, hash of the \emph{lowest uncertified block} $B_{u}$ that it carries. For example, in Figure~\ref{fig:snake}, $\Primary{2}$ proposes $B_{1,2}$ extending a \emph{New-Slot} certificate $\Certificate{3,1}$ and carrying $B_{4,1}$ ($B_u$) (Lines~\ref{slotleader:condition2}-\ref{slotleader:endcondition2}).

\begin{figure}[t]
\begin{myprotocol}

        \TITLE{Local state}{replica $R$}
        \STATE $\Certificate{s_{lp}, v_{lp}}$: the highest known certificate formed in view $v_{lp}$, slot $s_{lp}$
        \STATE $s, v$: the current slot and view 
        \STATE $B_h$: the highest voted block with hash $H_h$.
        \SPACE

        \TITLE{Leader role}{running at leader $\Primary{v}$}

        \EVENT{Upon \MName{pacemaker.EnterView()}}
            \STATE Keep updating $\Certificate{s_{lp}, v_{lp}}$ while receiving $\n-\f$ \MName{NewView} messages

            \STATE Wait until (1) formed a \emph{NEW-VIEW} certificate $\Certificate{s_w, w}$ \textbf{or} 
            (2) received $\n$ \MName{NewView} messages \textbf{or} \\
            (3) \MName{pacemaker.ShareTimer($v$)} \textbf{or} 
            (4) received $\n-k, 1\le k\le \f$, \MName{NewView} messages, but there are fewer than $\f{+}1{-}k$ votes for any slot higher than $(s_{lp}, v_{lp})$ 

            \IF{$\Primary{v}$ has not proposed $B_{1,v}$}\label{slot:has-not-proposed}
            \STATE Let $B_{1,v}$ be a block of client transactions yet to be proposed 
            \IF{(1) is satisfied}\label{slotleader:condition1}
            \STATE Broadcast $m = \SignMessage{\Message{Propose}{B_{1,v}, 1, v, \Certificate{s_w, w}, \bot}}{\Primary{v}}$ \label{slotleader:endcondition1}
            \ELSE \label{slotleader:condition2}
            \STATE $B_u \gets$ the lowest uncertified block that extends $\Certificate{s_{lp}, v_{lp}}$
            \STATE Broadcast $m = \SignMessage{\Message{Propose}{B_{1,v}, 1, v, \Certificate{s_{lp}, v_{lp}}, H_u}}{\Primary{v}}$ \label{slotleader:endcondition2}
            \ENDIF
            \ENDIF
        \ENDEVENT
        \SPACE

        \EVENT{Received $\n-\f$ \MName{NewView} messages with \emph{NEW-VIEW} signature-shares of $\Certificate{s_w,w}, w < v$}
            \STATE $\Certificate{s_w,w} \gets$ \MName{CreateNewViewThresholdSign}$(\n{-}\f ~\text{distinct}~\Share{h}{}~\text{shares}, fv = v)$
            \label{slot:create-newview-threshold} 
        \ENDEVENT
        \SPACE

        \EVENT{Received $\n-\f$ \MName{NewSlot} messages with \emph{NEW-SLOT} signature-share of $\Certificate{s,v}$} \label{slotleader:newslot}
            \STATE $\Certificate{s,v} \gets$ \MName{CreateNewSlotThresholdSign}$(\n-\f ~\text{distinct}~\Share{R}{}~\text{shares})$
            \STATE Let $B_{s+1,v}$ be a block of client transaction yet to be proposed 
            \STATE Broadcast $m = \SignMessage{\Message{Propose}{B_{s+1,v}, s+1, v, \Certificate{s, v}, \bot}}{\Primary{v}}$
            \label{slot:create-newslot-threshold} 
        \ENDEVENT

        \SPACE

        \EVENT{Received from a \emph{trusted} leader $\Primary{v-1}$ a \MName{NewView} message with a certificate formed in view $v-1$} \label{trusted:receive-previous-leader-newview-1} 
            \STATE Propose $B_{1,v}$ as in Lines~\ref{slot:has-not-proposed}-\ref{slotleader:endcondition2}
            \label{trusted:receive-previous-leader-newview-2} 
        \ENDEVENT
        \SPACE

        \EVENT{Received a \MName{Reject} message with $\Certificate{s_{v-1}^*,v-1}$} 
            \IF{received from $\Primary{v-1}$ a \MName{NewView} message with a lower certificate that is formed in view $v-1$} \label{trusted:receive-reject}
            \STATE Mark $\Primary{v-1}$ as \emph{distrusted}.
            \ENDIF
            \label{slot:create-newslot-threshold} 
        \ENDEVENT

    \end{myprotocol}
    \caption{Additional Local State and Leader Role in Streamlined HotStuff-1 with Slotting.}
    \label{alg:slot-steamlined-hs1-primary}
\end{figure}

\begin{figure}[t]
\begin{myprotocol}
    \FUNCTION{SafeSlot}{$s, v, \Certificate{s_w,w}, H_u$}
    \STATE Fetch the carried block $B_u$ of \textbf{non-empty} hash $H_u$
    \IF{$s=1$ \textbf{and} $\Certificate{s_w,w}$ is a \emph{NEW-VIEW} certificate \textbf{and} $\Certificate{s_w,w}.fv = v\;$ \COMMENT{Case 1}}
    \label{slot:case1}
    
        \RETURN true
    
    \ELSIF{$s=1$ \textbf{and} $\Certificate{s_w,w}$ is a \emph{NEW-VIEW} certificate \textbf{and} $\Certificate{s_w,w}.fv < v$ \textbf{and} $B_u.slot=1$ \textbf{and} $B_u.view = \Certificate{s_w,w}.fv$ \COMMENT{Case 2}}
    \label{slot:case2}
    \RETURN true
        
    \ELSIF{$s=1$ \textbf{and} $\Certificate{s_w,w}$ is a \emph{NEW-SLOT} certificate \textbf{and} $B_u.slot=s_w+1$ \textbf{and} $B_u.view = w$\COMMENT{Case 3}}
    \label{slot:case3}
    \RETURN true

    \ELSIF{$s > 1$ \textbf{and} $\Certificate{s_w,w}$ is a \emph{NEW-SLOT} certificate \textbf{and} $s_w = s-1$ \textbf{and} $w=v$ \COMMENT{Case 4}}
    \label{slot:case4}
        \RETURN true
    \ENDIF
    \RETURN false \label{slot:endcheck}
    \ENDFUNCTION
    \SPACE
        \TITLE{Backup role}{running at each replica $R$ (including leader)}
        \EVENT{Received $\SignMessage{\Message{Propose}{B_{s,v}, \Certificate{s_w,w}, H_u}}{\Primary{v}}$} \label{slot:recv-propose}
            \IF{$\Certificate{s_w,w}$ extends $\Certificate{s_w-1,w}$ \COMMENT{commit-rule-case1}} \label{slot:commit-rule-case1}
                \STATE Execute all transactions up to (incl.) $B_{s_w-1,w}$, add result to global-ledger and respond to clients  
            \SPACE

            \ELSIF{$s_w = 1$ \AND $\Certificate{s_w,w}$ extends $\Certificate{s_{w-1},w-1}$
            \COMMENT{commit-rule-case2}} \label{slot:commit-rule-case2}
                \STATE Execute all transactions up to (incl.) $B_{s_{w-1},w-1}$, add result to global-ledger and respond to clients  
            \ENDIF
            \SPACE
            \IF{($s = s_w+1$ \AND $v = w$) \OR ($s = 1 $ \AND $v = w+1$) \COMMENT{No-Gap rule} \\\hspace{4mm}\AND predecessor of $\Certificate{s_w,w}$ is in global-ledger  \COMMENT{\strongspec{} rule}}\label{slot:prefixspeculation}
                \IF{local-ledger state conflicts with $B_{s_w,w}$} \label{slot:conflict-cert}
                    \STATE Rollback local-ledger to the common ancestor \label{slot:rollback}
                \ENDIF
                \STATE Execute all transactions in $B_{s_w,w}$ speculatively, add the result to local-ledger and send the client a response 
            \ENDIF
            \SPACE

            \IF{\textsc{SafeSlot($s,v,\Certificate{s_w, w}, H_u$)}\label{line:consecutiveslot}
            }
                \STATE $\Share{R}{} \gets$ \MName{CreateThresholdShare}$(\Certificate{s_{lp}, v_{lp}}, Hash(B_{s,v}), H_u, New{-}Slot)$  \label{slot:create-newslot-share}
                \STATE Send $\SignMessage{\Message{NewSlot}{s, v,\Certificate{s_{lp},v_{lp}},\Share{R}{}}}{R}$ to $\Primary{v}$ \label{slot:send-newslot}
            \ELSE
            \STATE Send $\SignMessage{\Message{Reject}{s, v,\Certificate{s_{lp},v_{lp}}}}{R}$ to $\Primary{v}$ \label{trusted:send-reject}
            \ENDIF
            
            \SPACE

            \STATE $s \gets s+1$  \COMMENT{disable voting for slot $s$} 
        \ENDEVENT
        \SPACE

        \EVENT{Upon timeout} \label{slot:timeout}  
            %
            \STATE $\Share{h}{} \gets$ \MName{CreateThresholdShare}$(\Certificate{s_{lp}, v_{lp}}, H_h, New{-}View)$  \label{slot:create-newview-share} 
            \STATE Send $\SignMessage{\Message{NewView}{v+1,\Certificate{s_{lp},v_{lp}}, H_h, \Share{h}{}}}{R}$ to $\Primary{v+1}$  \label{slot:send-newview}
            \STATE $v \gets v+1$, $s \gets 1$\COMMENT{disable voting for view $v$}
            \STATE Call \MName{pacemaker.completedView}$()$ \label{slot:switch-view} 
        \ENDEVENT

    \end{myprotocol}
    \caption{Backup Role in Streamlined \HSOne{} + Slotting.}
    \label{alg:slot-steamlined-hs1-replica}
\end{figure}


{\em NewSlot:} For slot $(s, v)$, where $s > 1$, the leader $\Primary{v}$ awaits well-formed \MName{NewSlot} 
messages from at least $\n{-}\f$ replicas voting for $B_{s-1, v}$. Once it collects $\n{-}\f$ votes, it combines the \emph{New-Slot} signature-shares to create a \emph{New-Slot} certificate $\Certificate{s-1,v}$. After forming $\Certificate{s-1,v}$, $\Primary{v}$ proceeds 
to propose slot $B_{s, v}$ including $\Certificate{s-1,v}$ (Lines~\ref{slotleader:newslot}-\ref{slot:create-newslot-share}).



{\bf ProposeVote.}
Upon receiving a proposal $B_{s,v}$, a replica $R$ checks whether any of the following cases are satisfied—based on the slot $s$, the certificate $\Certificate{s_w, w}$, and the block $B_u$ with hash $H_u$—before voting (see Figure~\ref{alg:slot-steamlined-hs1-replica}, Lines~\ref{slot:case1}–\ref{slot:endcheck}). \Add{For first-slot proposals, Case 1 provides a valid proof of “no tail-forking” in way (i), while Cases 2 and 3 provide a valid proof in way (ii).}

Case 1: $s=1$; $B_{1,v}$ extends a \emph{New-View} certificate $\Certificate{s_w, w}$ such that $\Certificate{s_w, w}.fv = v$.

Case 2: $s=1$; $B_{1,v}$ extends a \emph{New-View} certificate $\Certificate{s_w, w}$ such that $\Certificate{s_w, w}.fv < v$; and $B_{1,v}$ carries $B_u$ such that $B_u.slot=1$ and $B_u.view = \Certificate{s_w,w}.fv$.

Case 3: $s=1$; $B_{1,v}$ extends a \emph{New-Slot} certificate $\Certificate{s_w, w}$ and carries $B_{s_w+1, w}$.

Case 4: $s>1$; $B_{s,v}$ extends a \emph{New-Slot} certificate $\Certificate{s_w, w}$ such that $s_w = s-1$ and $w = v$.

Then, $R$ checks if its highest certificate $\Certificate{s_{lp}, v_{lp}}$ is lexicographically not greater than $\Certificate{s_w, w}$. If so, $R$ sends a \MName{NewSlot} message containing a \emph{New-Slot} signature-share of $B_{s,v}$ (Line~\ref{slot:create-newslot-share}).




{\bf Slot-change.}
There is no timer for individual slots within a view: given a view $v$, a replica exits slot $s$ upon receiving a well-formed leader proposal for slot $(s, v)$, which extends $\Certificate{s-1, v}$.

{\bf View-change.}
A lack of progress is detected at the view level (not at the slot level). When the timer for view 
$v-1$ expires, a replica $R$ exits view $v-1$; 
$R$ uses the pacemaker to synchronize entering to view $v$ and sends a \MName{NewView} message containing $\Certificate{s_{lp},v_{lp}}$, its highest certificate, $H_h$, hash of its highest voted block, and $\delta_h$, a \emph{New-View} signature share for $\Certificate{s_{lp},v_{lp}}$ and $H_h$ (Lines~\ref{slot:timeout}-\ref{slot:switch-view}).

{\bf Commit Rule.}
The same as the streamlined design without \emph{slotting}, \PHSOne{} with \emph{slotting} has only one commit rule: 
replicas follow the prefix commit rule (Definition \ref{def:consecutive-commit-rule}) to add a transaction to the global-ledger. 

However, as we form a two-dimensional chain with \emph{slotting}, there are two different cases when a replica $R$ learns a new certificate $\Certificate{s_w,w}$ and commits the block extended by $\Certificate{s_w,w}$:
(1) $s_w > 1$: commits block $B_{s_w-1,w}$ if $\Certificate{s_w,w}$ extends $\Certificate{s_w-1,w}$. (Line~\ref{slot:commit-rule-case1})  (2) $s_w = 1$: commits block $B_{s_{w-1},w-1}$ if $\Certificate{s_w,w}$ extends $\Certificate{s_{w-1},w-1}$ (Line~\ref{slot:commit-rule-case2}).

Of special note are the \emph{uncertified carry blocks} in the first-slot blocks, that are viewed as a part of the first-slot blocks. If a first-slot block $B_{1,v}$ contains $H_u$ of a carry block $B_{u}$, then $B_{u}$ gets committed only when $B_{1,v}$ is committed.

{\bf Speculation.} 
Replicas may speculate on a block $B_{s_w,w}$ when it satisfies
the Prefix Speculation Rule and No-Gap Rule. That is, a replica $R$ can speculate on a block $B_{s_w,w}$ upon receiving a proposal $B_{s,v}$ carrying $\Certificate{s_w,w}$ if the prefix of $B_{s_w,w}$ is committed and $B_{s_w,w}$ is from the immediately preceding slot (Line~\ref{slot:prefixspeculation}), i.e.,
(1) $s = s_w+1, v = w$; or 
(2) $s = 1, v = w+1$. 



\subsection{Tolerance to Tail-Forking} \label{sec:tailfork}

Now we show how the \HSOne{} with slotting mitigates tail-forking attacks. We denote by $B_{s-1,v}, B_{s,v}$, the last two slots of view $v$ with a correct leader $\Primary{v}$, where $B_{s, v}$ extends $B_{s-1,v}$. It is guaranteed that if $\Primary{v}$ proposed at least two slots and at least $\f+1$ correct replicas have voted for its last slot $B_{s,v}$, then $B_{s,v}$ and all preceding slots in view $v$ are protected from \emph{tail-forking attacks}.


That is, if at least $\f{+}1$ correct replicas have voted for $B_{s,v}$, it becomes impossible to form a \emph{New-View} certificate $\Certificate{s{-}1,v}$, as those $\f{+}1$ correct replicas will vote for $B_{s,v}$ rather than $B_{s{-}1,v}$ in their \MName{NewView} messages. Consequently, $B_{1,v{+}1}$ can extend either a \emph{New-Slot} certificate $\Certificate{s{-}1,v}$ or a \emph{New-View} certificate $\Certificate{s,v}$. If it extends the \emph{New-Slot} $\Certificate{s{-}1,v}$, then by Case 3 of \MName{ProposeVote} phase, it must \emph{carry} $B_{s,v}$; otherwise, it extends the \emph{New-View} certificate $\Certificate{s,v}$. In either case, $B_{s,v}$ and all preceding slots in view $v$ are not \emph{tail-forked}.

\subsection{Advancing at Network Speed with Trusted Previous Leaders}\label{sec:slot_responsiveness}

Generally, leaders of BFT consensus must guarantee they extend a highest certificate that all honest replicas will accept (for liveness).
A hallmark of protocols in the \HS{} family, often referred to as \textit{(optimistic) responsiveness}, is allowing the protocol to advance \textit{at network speed} unless there are faults. In particular, in \HS/\HSTwo, the leader replacement regime ensures that (after GST), leaders learn the highest certificate without waiting for the pre-determined maximal network delay $\Delta$, unless there is a fault.
%


\PHSOne{} with \emph{slotting} introduces a new challenge: $\Primary{v}$ does not know in advance the highest slot $s$ proposed in view $v{-}1$, since each leader attempts to propose as many slots as possible before its view expires, and the number of slots per view is adaptive. If a correct leader $\Primary{v-1}$ fails to broadcast its final slot to at least $\n{-}\f$ well-behaving replicas before their view timers expire, the next leader $\Primary{v}$ may be unable to form a \emph{New-View} certificate and must wait for an $O(\Delta)$ delay to receive \MName{NewView} messages from all correct replicas.

To avoid this unintended $O(\Delta)$ delay between two correct leaders, we introduce the notion of \emph{trusted} and \emph{distrusted} previous leaders. Initially, each leader trusts its previous leader in the rotation. Upon receiving a \MName{NewView} message from a \emph{trusted} previous leader $\Primary{v{-}1}$ that includes a certificate formed in view $v-1$ (Figure~\ref{alg:slot-steamlined-hs1-primary}, Line~\ref{trusted:receive-previous-leader-newview-1}), i.e., a \emph{New-Slot} certificate of view $v-1$ or a \emph{New-View} certificate with $fv=v-1$, $\Primary{v}$ immediately proposes its first-slot extending $\Certificate{s, v{-}1}$. This is safe because no correct replica can hold a higher certificate than that of a correct $\Primary{v{-}1}$ when exiting view $v{-}1$, assuming a certificate was formed in view $v-1$. The trusted/distrusted mechanism thus enables $\Primary{v}$ to propose its first-slot block at network speed, avoiding unnecessary delays.

However, a Byzantine previous leader $\Primary{v{-}1}$, initially \emph{trusted} by $\Primary{v}$, may conceal the highest certificate it has formed and sent to other correct replicas. This can cause $\Primary{v}$'s first-slot proposal to be rejected by a correct replica $R$ that has already received the higher certificate. Upon rejection, $R$ sends a \MName{Reject} message to $\Primary{v}$ containing its highest certificate (Figure~\ref{alg:slot-steamlined-hs1-replica}, Line~\ref{trusted:send-reject}). If $\Primary{v}$ had previously received a \MName{NewView} message from $\Primary{v{-}1}$ containing a lower certificate formed in view $v-1$ (Figure~\ref{alg:slot-steamlined-hs1-primary}, Line~\ref{trusted:receive-reject}), it then marks $\Primary{v{-}1}$ as \emph{distrusted}.
In future views where $\Primary{v}$ becomes leader again, it no longer trusts $\Primary{v{-}1}$ and follows the four conditions described in \S\ref{ssec:slotting} when entering the view. As a result, each malicious leader $\Primary{v{-}1}$ can conceal its highest certificate at most once after \emph{GST}, without compromising liveness.


\Add{
If the previous leader is \emph{distrusted}, the four conditions for proposing the first slot ensure that, after \emph{GST}, a correct leader $\Primary{v}$ can learn the highest certificate known to any correct replica.
If condition~(\ref{condition1}) holds, then at least $\f{+}1$ correct replicas did not vote for any slot higher than the formed certificate $\Certificate{s_w, w}$, implying that no higher certificate could exist.
If condition~(\ref{condition2}) or~(\ref{condition3}) holds, the Pacemaker guarantees that $\Primary{v}$ receives \MName{NewView} messages from all correct replicas, thereby acquiring the highest certificate.
Under condition~(\ref{condition4}), the highest votes contained in the \MName{NewView} messages reveal that no higher certificate could have been formed.
}

\Arxiv{
If some correct replica holds a certificate $\Certificate{s^*,v^*}$ higher than the leader's $\Certificate{s_{lp},v_{lp}}$, then at least $\f+1$ correct replicas have voted for $B_{s^*,v^*}$ and will vote for a block not lower than $B_{s^*,v^*}$ in the \MName{NewView} messages sent to the leader. While processing the \MName{NewView} messages, if condition (\ref{condition1}) or (\ref{condition4}) is satisfied, then no block higher than $\Certificate{s_{lp}, v_{lp}}$ can get more than $\f$ correct-replica votes through the \MName{NewView} messages, thus such $\f+1$ correct replicas do not exist; If condition (\ref{condition2}) or (\ref{condition3}) is satisfied, then the higher certificate will be received by the leader, as the pacemaker protocol guarantees that after \emph{GST}, all \MName{NewView} messages from correct replicas can arrive at the leader by \MName{ShareTimer($v$)
}.
}

\if{false}

{\bf Adaptive Slotting Semantics.}
The leader in each view may adopt count-based or time-based window semantics when proposing slots. 
In the count-based variant, each leader is expected to propose a predetermined set of slots, while 
in the time-based variant, a predetermined timeout period is set in which the leader may propose as many 
slots as possible adaptively. Since the number of slots is not known \textit{a priori} when adopting 
the time-based semantics, the correct leader will tag its last slot as the final slot to invoke the 
pacemaker in correct replicas to advance their views.

\DM{omit the rest below here.}
$\Primary{v}$ uses its highest locked certificate (say $\Certificate{j, v_{lp}}$) to extend 
its new proposal. In either case, the leader broadcasts its new proposal as a $\MName{Propose}$ message 
to all the replicas. This proposal includes the view number $v$, block $B_{i,v}$ of client transactions 
yet to be proposed, hash $h_{i,v}$ of $B_{i,v}$, and the certificate it extends.

{\bf Execute and Ledger Update.}
On receiving the \MName{Propose} message $m$ from $\Primary{v}$ for slot $i$, a replica $R$ 
compares the views of the certificate $\Certificate{k,w}$ in $m$ and and $R$'s highest locked certificate 
$\Certificate{j, v_{lp}}$. If $w < v_{lp}$, $R$ ignores the message. Otherwise, $R$ performs as follows.

\begin{enumerate}[nosep,wide]
    \item Updates $\Certificate{j,v_{lp}}$ to $\Certificate{k,w}$; sets $v_{lp} = w$.
    
    \item If $\Certificate{k,w}$ conflicts with $\Certificate{j,v_{lp}}$ and $R$ had speculatively 
    executed transaction at $v_{lp}$, then $R$ rollbacks its state. 
    
    \item If $\Certificate{k,w}$ is a strong certificate: $\Certificate{k,w}$ extends a 
    certificate $\Certificate{k-1,w}$ if $k > 1$ or $\Certificate{s,w-1}$ if $k = 1$, where $w=v-1$, 
    then $R$ {\em speculatively executes} the transactions corresponding to $\Certificate{k,w}$ 
    once it has executed transactions (yet to be executed) for certificates that $\Certificate{k,w}$ 
    extends. The order for executing these transactions is the same as described earlier in this section.

    \item Replies to the respective clients.
    
    \item Adds an entry for $\Certificate{k,w}$ in its local log.
\end{enumerate}

{\bf Commit.}
Post this, $R$ {\em adds} an entry for the certificates that $\Certificate{k,w}$ extends 
to the global ledger. 

{\bf NewView.}
Concurrently, $R$ creates a \MName{ProposeVote} message, which includes a threshold signature 
share $\Share{R}{P}$ for $m$ and hash of $\Certificate{j, v_{lp}}$; $R$ sends this message to 
the leader for the next slot.

\fi

%
%

\section{Evaluation}\label{sec:eval}
Our evaluation aims to answer the following:
\begin{enumerate}[wide,nosep]
    \item Scalability of \sysname{}: throughput and latency with a varying number of replicas and number of transactions in a batch.
    \item Impact of $\f$ additional required responses on \sysname{}.
    \item Impact of leader-slowness, tail-forking, and rollbacks.
\end{enumerate}

{\bf Setup.}
We use c3.4xlarge AWS machines: $16$-core Intel Xeon E5-2680 v2 (Ivy Bridge) processor, 
$\SI{2.8}{\giga\hertz}$ and $\SI{30}{\giga\byte}$ memory. We deploy up to $64$ machines for 
replicas. Each experiment runs for $120$ seconds.
We employ {\em batching} in all our experiments with a default batch size of $100$ and 
mention specific sizes when necessary.

{\bf Implementation.}
We implement all the protocols in \RDB{} (incubating)~\cite{apacheresdb}; C++20 code with Google 
Protobuf v$3.10.0$ for serialization and NNG v$1.5.2$ for networking. \RDB{} is an optimized 
blockchain framework that provides APIs to implement a new consensus protocol. As threshold signature 
algorithms are expensive and can quickly bottleneck the computational resources, the leader sends a 
list of $\n-\f$ digital signatures (from distinct replicas) as a certificate.

{\bf Baselines.}
We compare streamlined \sysname{} against two other comparable streamlined protocols:
\begin{enumerate}[wide,nosep]
    \item {\bf \HS{}.} 
    First streamlined \BFT{} consensus protocol; requires $7$ half-phases to reach consensus 
    on a client transaction (total $9$ half-phases including client request and response).

    \item {\bf \HSTwo{}.}
    Optimized \HS{} variant that requires $5$ half-phases for consensus (total $7$ half-phases).
\end{enumerate}

As for \sysname{}, we implement two versions of it:
\begin{enumerate}[wide,nosep]
    \item {\bf \HSOne{}.}
    Streamlined \BFT{} consensus protocol with speculative execution that 
    requires $3$ half-phases for speculative response (total $5$ half-phases).

    \item {\bf \HSOne{} (with Slotting).}
\end{enumerate}

{\bf Workloads.}
We use two workloads: YCSB~\cite{blockbench} and TPC-C~\cite{tpcc}:
\begin{enumerate}[wide,nosep]
    \item {YCSB.} Key-value store write operations that access a database of 600k records.

    \item {TPC-C.} Online transaction processing (OLTP) operations that access a database of 260k records, simulating a complex warehouse and order management environment.
\end{enumerate}
Unless explicitly stated, we use YCSB as the default workload.

{\bf Metrics.}
We focus on two metrics:

(1) {\em Throughput} -- the maximum number of transactions per second for which the system completes consensus.

%
(2) {\em Client Latency} -- the average duration between the time a client sends a transaction to the time 
the client receives a matching quorum of responses ($\f{+}1$ for \HS{}/\HSTwo{} and $\n{-}\f$ for \sysname{}) for that transaction. 

\begin{figure*}[t]
    \centering
    \scalebox{0.5}{\ref{mainlegend4}}\\[5pt]
    \setlength{\tabcolsep}{4pt}
    \renewcommand{\arraystretch}{1.0}
    \begin{tabular}{cccc}
        \begin{tabular}{@{}c@{}}
            \throughputgraphb{\dataTputb}{(a) Scalability}{\axisnodes}{Throughput (txn/s)}{\axisticksnodesb} \\[4pt]
            \geothroughputgraph{\datageo}{(e) Geo-Scale + YCSB}{\axisregions}{Throughput (txn/s)}{\axisticksregions}
        \end{tabular}
        &
        \begin{tabular}{@{}c@{}}
            \latencygraph{\dataClientLat}{(b) Scalability}{\axisnodes}{Client Latency (ms)}{\axisticksnodesb} \\[4pt]
            \geolatencygraph{\responselatgeo}{(f) Geo-Scale + YCSB}{\axisregions}{Client Latency (s)}{\axisticksregions}
        \end{tabular}
        &
        \begin{tabular}{@{}c@{}}
            \batchsizegraph{\dataBatchSize}{(c) Batching}{\axisbatches}{Throughput (txn/s)}{100, 1000, 2000, 5000, 10000} \\[4pt]
            \geothroughputgraph{\datageotpcc}{(g) Geo-Scale + TPC-C}{\axisregions}{Throughput (txn/s)}{\axisticksregions}
        \end{tabular}
        &
        \begin{tabular}{@{}c@{}}
            \batchsizeLatgraph{\ClientLatBatchSize}{(d) Batching}{\axisbatches}{Client Latency (ms)}{100, 1000, 2000, 5000, 10000} \\[4pt]
            \geolatencygraph{\responselatgeotpcc}{(h) Geo-Scale + TPC-C}{\axisregions}{Client Latency (s)}{\axisticksregions}
        \end{tabular}
    \end{tabular}
    \caption{Scalability Plots.}
    \label{fig:throughput}
\end{figure*}

\subsection{Scalability}
\label{sssec:tput_rep}

{\bf Impact of the number of replicas} In Figures~\ref{fig:throughput} (a) and (b), we present various system metrics as a function of the 
number of replicas; we increase the number of replicas from $\n=4$ to $\n=64$.

As expected, an increase in the number of replicas causes a proportional decrease in the throughput for 
all the protocols due to an $O(\n)$ increased message complexity, which decreases available bandwidth 
and increases the computational work at each replica. \HSOne{}, with or without slotting, yields the same throughput as \HS{}/\HSTwo{}
because the message complexity remains the same for all the streamlined protocols. 

An increase in the number of replicas also causes a proportional increase in the client 
latency for all the protocols due to an $O(\n)$ increased message complexity, which increases the time 
duration for a leader to collect a quorum of threshold shares and to form a certificate. 
Moreover, each client needs to wait longer for a larger quorum of messages to arrive.
This implies that \HSOne{} clients should incur higher latency as they must wait for $\f$ more 
responses. However, \HSOne{} yields lower latency because speculation guarantees 
an \early{}. \HSOne{}, with or without slotting, yields $41.5 \%$ and $24.2\%$ (for small setups) and 
$38.5\%$ and $22.7\%$ (for large setups) less client latency in comparison to \HS{} and \HSTwo{}.

{\bf Impact of Batch Size}
Next, in Figures~\ref{fig:throughput} (c) and (d), we increase the number of transactions per batch 
(batch size) from $100$ to $10000$ and run consensus among $\n = 32$ replicas.

For all protocols, increasing the batch size improves throughput until either bandwidth or compute resources are saturated, beyond which throughput tapers off.
The throughput gain at smaller batch sizes is due to reduced consensus overhead and fewer messages being processed.
At larger batch sizes (around $5000$), all protocols become compute-bound before reaching bandwidth saturation, as the benefits of reduced consensus overhead are offset by the increased cost of proposing (for leaders) and processing (for replicas) larger batches.
In contrast, client latency increases with batch size, as proposing and processing larger batches takes more time in each view.

%
%

{\bf Geo-Scale Scalability}
In Figures~\ref{fig:throughput} (e–h), we deploy replicas across the globe, varying the number of geographical regions from $2$ to $5$—North Virginia, Hong Kong, London, São Paulo, and Zurich—and uniformly distribute $\n=32$ replicas across these regions. These experiments use both the YCSB and TPC-C benchmarks. We observe that all protocols exhibit similar trends across both benchmarks, as high inter-regional round-trip times limit throughput and increase latency.

As the number of regions increases from $2$ to $5$, all protocols experience up to a $59\%$ drop in throughput and a $159.4\%$ increase in latency. Nevertheless, the general trend remains consistent: \HSOne{} matches the throughput of other protocols while achieving the lowest latency.


\subsection{Impact of the $\f$ Additional Responses}
We now experimentally validate our claim: although \HSOne{} clients wait for $\f$ additional responses compared to \HS/\HSTwo{} clients, \HSOne{} always yields the lowest latency for clients.

\begin{figure*}[t]
    \centering
    \scalebox{0.5}{\ref{mainlegend4}}\\[3pt]
    \setlength{\tabcolsep}{2pt}  
    \renewcommand{\arraystretch}{1.0}
    \begin{tabular}{@{}c@{}c@{}c@{}c@{}c@{}}  
        \begin{tabular}{@{}c@{}}
            \networkdelaytesttput{\networkdelayonemstput}{(a) Inject $1$ms Delay}{Number of Impacted Replicas}{Throughput (txn/s)} \\
            \networkdelaytestlat{\networkdelayonemslat}{(f) Inject $1$ms Delay}{Number of Impacted Replicas}{Latency (ms)}
        \end{tabular}
        &
        \begin{tabular}{@{}c@{}}
            \networkdelaytesttput{\networkdelayfivemstput}{(b) Inject $5$ms Delay}{Number of Impacted Replicas}{Throughput (txn/s)} \\
            \networkdelaytestlat{\networkdelayfivemslat}{(g) Inject $5$ms Delay}{Number of Impacted Replicas}{Latency (ms)}
        \end{tabular}
        &
        \begin{tabular}{@{}c@{}}
            \networkdelaytesttput{\networkdelayfiftymstput}{(c) Inject $50$ms Delay}{Number of Impacted Replicas}{Throughput (txn/s)} \\
            \networkdelaytestlat{\networkdelayfiftymslat}{(h) Inject $50$ms Delay}{Number of Impacted Replicas}{Latency (ms)}
        \end{tabular}
        &
        \begin{tabular}{@{}c@{}}
            \networkdelaytesttput{\networkdelayfivehundredmstput}{(d) Inject $500$ms Delay}{Number of Impacted Replicas}{Throughput (txn/s)} \\
            \networkdelaytestlat{\networkdelayfivehundredmslat}{(i) Inject $500$ms Delay}{Number of Impacted Replicas}{Latency (ms)}
        \end{tabular}
        &
        \begin{tabular}{@{}c@{}}
            \networkdelaytesttput{\tworegiontput}{(e) Geographical Deployment}{Number of London Replicas}{Throughput (txn/s)} \\
            \networkdelaytestlat{\tworegionlat}{(j) Geographical Deployment}{Number of London Replicas}{Latency (ms)}
        \end{tabular}
    \end{tabular}
    \caption{Performance with Varying Network Conditions.}
    \label{fig:networkdelaytest}
\end{figure*}

{\bf Injecting Message Delay.}
We begin by evaluating the impact of delayed messages on client latency. This experiment demonstrates that even when more than $\f+1$ replicas experience high message delays, \HSOne{} clients do not incur increased latencies.
The setup is as follows:
(1) We deploy $\n=31$ replicas.
(2) Based on prior experiments, the client latencies for \HSOne{}/\HSTwo{}/\HS{} are approximately $5\ \mathrm{ms}$/$7\ \mathrm{ms}$/$9\ \mathrm{ms}$, respectively. We inject increasing message delays $\delta \in \{1\ \mathrm{ms},\ 5\ \mathrm{ms},\ 50\ \mathrm{ms},\ 500\ \mathrm{ms}\}$.
(3) We vary the number of impacted replicas $k \in \{0,\ \f,\ \f{+}1,\ \n{-}\f{-}1,\ \n{-}\f,\ \n\}$, i.e., $k = 0,\ 10,\ 11,\ 20,\ 21,\ 31$.
Figures~\ref{fig:networkdelaytest} (a–d) and (f–i) present the results of these experiments.

For all protocols, as the number of impacted replicas increases, latency increases and throughput decreases due to the delayed message transmission to and from these replicas. The impact is most pronounced when increasing from $k = \f$ (10) to $k = \f{+}1$ (11), as every certificate formed by the leader must now include at least one signature share from an impacted replica (since certificates require $\n{-}\f$ signatures). These results further support our claim that the primary bottleneck in these protocols lies in achieving consensus, rather than in responding to clients.

As the number of impacted replicas increases from $k = \n{-}\f{-}1$ (20) to $k = \n{-}\f$ (21), client latencies in \HS{} and \HSTwo{} increase sharply, whereas \HSOne{} shows only a moderate increase. This is because, when $k \ge \n{-}\f$, clients can receive at most $\f$ responses from non-impacted replicas, causing latency to be dominated by the slower, impacted replicas.

When $k \le \f$, \HSOne{} with \emph{slotting} yields better performance than all other protocols because slotting allows the non-impacted replicas to propose more blocks during their views.

{\bf Geographical Deployment.}
Next, we deploy $\n=31$ replicas across two geographically distant regions: North Virginia and London, with all clients located in North Virginia. We vary the number of replicas placed in London, denoted by $k \in \{0, \f, \f{+}1, \n{-}\f{-}1, \n{-}\f, \n\}$. The results, shown in Figures~\ref{fig:networkdelaytest}(k) and (l), illustrate the impact of increasing geographic separation among replicas.

When $k \le \f$ (10) or $k \ge \n{-}\f$ (21), leaders in North Virginia and London, respectively, can form certificates using votes from $\n{-}\f$ replicas within their own region. In contrast, when $k$ is between $\f+1$ (11) and $\n-\f-1$ (20), forming a certificate requires at least one vote from the remote region, leading to degraded throughput and latency. Performance is better when $k \le \f$ than when $k \ge \n{-}\f$ because most leaders are co-located with clients in North Virginia. When $k \le \f$ or $k \ge \n-\f$, \HSOne{} with \emph{slotting} outperforms other protocols, as slotting enables leaders with $\n{-}\f$ co-located replicas to propose more blocks per view.

\subsection{Failure Resiliency}
\label{ss:eval-failures}

\begin{figure*}[t]
    \centering
    \scalebox{0.5}{\ref{falegend}}\\[5pt]
    \setlength{\tabcolsep}{4pt}
    \renewcommand{\arraystretch}{1.0}
    \begin{tabular}{cccc}
        \begin{tabular}{@{}c@{}}
            \nonresponsivegraph{\dataNR}{(a) Leader slowness (timer 10ms)}{Number of Slow Leaders}{Throughput (txn/s)}{\axistickslows} \\[4pt]
            \forktailgraph{\dataFA}{(e) Tail-forking}{Number of Faulty Leaders}{Throughput (txn/s)}{\axistickslows}
        \end{tabular}
        &
        \begin{tabular}{@{}c@{}}
            \nonresponsivelatgraph{\latNR}{(b) Leader slowness (timer 10ms)}{Number of Slow Leaders}{Client Latency (ms)}{\axistickslows} \\[4pt]
            \forktailgraph{\latFA}{(f) Tail-forking}{Number of Faulty Leaders}{Client Latency (ms)}{\axistickslows}
        \end{tabular}
        &
        \begin{tabular}{@{}c@{}}
            \nonresponsivegraphb{\dataNRB}{(c) Leader slowness (timer 100ms)}{Number of Slow Leaders}{Throughput (txn/s)}{\axistickslows} \\[4pt]
            \forktailgraph{\dataRB}{(g) Rollback}{Number of Faulty Leaders}{Throughput (txn/s)}{\axistickslows}
        \end{tabular}
        &
        \begin{tabular}{@{}c@{}}
            \nonresponsivelatgraphb{\latNRB}{(d) Leader slowness (timer 100ms)}{Number of Slow Leaders}{Client Latency (ms)}{\axistickslows} \\[4pt]
            \forktailgraph{\latRB}{(h) Rollback}{Number of Faulty Leaders}{Client Latency (ms)}{\axistickslows}
        \end{tabular}
    \end{tabular}
    \caption{Impact of varying the number of faulty replicas (leader slowness, tail-forking, and rollback).}
    \label{fig:fault}
\end{figure*}

{\bf Leader slowness phenomenon.}

We now study the impact of leader slowness (\S\ref{s:slotting}) on streamlined protocols by varying the number of slow leaders from $0$ to $\f$, with $\n=32$ replicas, a batch size of $100$, and two timeout settings: \SI{10}{ms} and \SI{100}{ms}. A slow leader does not propose until the end of its view duration. Figures~\ref{fig:fault}(a)–(d) present the results.

Slow leaders degrade throughput and client latency in all protocols except \HSOne{} with \emph{slotting}. In \HSOne{}, each leader can propose multiple slots, eliminating delays associated with leader slowness. Moreover, the larger the timeout period, the more batches a leader can propose, further improving performance. For example, with a timeout of \SI{10}{ms}, \HSOne{} (with slotting) experiences only a $1.8\%$ and $28.7\%$ drop in throughput and a $0.9\%$ and $18.5\%$ increase in latency with $1$ and $\f=10$ slow leaders, respectively. In contrast, other protocols suffer $14.5\%$ and $63.5\%$ lower throughput and $18.7\%$ and $2.8\times$ higher latency under the same conditions.
Similarly, with a timeout of \SI{100}{ms}, \HSOne{} incurs $3.9\%$ and $34.4\%$ lower throughput and $5.7\%$ and $27.1\%$ higher latency with $1$ and $\f=10$ slow leaders, while other protocols see throughput drop by $63.4\%$ and $94.5\%$, and latency increase by $2.81\times$ and $19\times$, respectively.

{\bf Tail-forking attack.}
Similar to the leader slowness phenomenon, the tail-forking attack seeks to increase system latency by preventing proposals from correct leaders from being committed (\S\ref{s:slotting}). In this experiment, we vary the number of faulty leaders from $0$ to $\f$, using $\n=32$ replicas and a batch size of $100$. As shown in Figures~\ref{fig:fault}(e) and (f), a faulty leader in view $v$ ignores the certificate from the proposal in view $v{-}1$ and instead extends its proposal from the certificate of view $v{-}2$.

As before, faulty leaders degrade the performance of all protocols except \HSOne{} with \emph{slotting}. In \HSOne{} with \emph{slotting}, each leader can propose multiple batches, and a faulty leader can at most suppress the final slot, mitigating the impact of the attack. \sysname{} with \emph{slotting} demonstrates greater resilience, particularly with longer timeout periods. For instance, with $\f=10$ faulty leaders, \HSOne{} (with slotting) shows only a $4.1\%$ and $1.4\%$ reduction in throughput under timeout settings of \SI{10}{ms} and \SI{100}{ms}, respectively—compared to the no-failure case. In contrast, other protocols suffer a $31.6\%$ drop in throughput under the same conditions.
Similarly, \HSOne{} experiences minimal change in latency, while other protocols exhibit up to a $45.3\%$ increase in client latency with $\f=10$ faulty leaders.

{\bf Rollback.} 
In \HSOne{}, speculation on uncommitted transactions may require replicas to roll back speculated transactions. In Figures~\ref{fig:fault}(g)–(h), we vary the number of faulty leaders from $0$ to $\f$ and allow each to force up to $\f$ correct replicas to roll back transactions, using $\n=32$ replicas and a batch size of $100$. Notably, in \HSOne{} with \emph{slotting}, a faulty leader $\Primary{v}$ can only force rollbacks of the last slot in the preceding view $v{-}1$; it cannot skip any slot in its own view.
Faulty leaders degrade throughput and latency in \HSOne{} without slotting. With $\f=10$, \HSOne{} without slotting suffers a $38.1\%$ drop in throughput and a $35.8\%$ increase in latency relative to the no-failure case. In contrast, rollback attacks have minimal impact on \HSOne{} with slotting.

\section{Related Work}\label{sec:related}

Extensive literature exists on consensus, with numerous studies 
(e.g.,~\cite{wild,scaling,untangle,leaderless-consensus,next700bft,zyzzyvaj,sbft,ardagna2020blockchain,xft,
loghin2022blockchain,blockchain-info-share,lineagechain,bft-forensics,ftn-consensus,gem2tree,bedrock,rbft,serverlessbft,sharper,honeybadger}) 
focused on enhancing consensus systems~\cite{kogias2016enhancing,borealis,shadoweth,occlum,experimental-bft-improv,bft-to-cft,prestigebft,geobft,kuhring2021streamchain,disth,disc-mbft,resdb-icdcs,ringbft,ccf-heidi,eesmr-bft,chemistry,orthrus,ava}. 

\textbf{Speculation.}
Protocols belonging to the \PBFT{} family~\cite{zyzzyvaj,zyzzyva-unsafe,sbft} have explored an \textit{optimistic fast-path} approach to speculation. 
Unfortunately, it works only in fault-free runs and requires a quadratic fallback mechanism. 
Several papers try to eliminate the dependence on the fast-path, but under leader failures, 
they also require quadratic fallback mechanisms~\cite{poe, poe-cc}. 
Exposing the \strongspec{} dilemma and suggesting a rule to resolve it may benefit all of these.

\textbf{Rotational Leader.}
The \HS{} family of protocols reduces leader-replacement communication costs to linear, 
enabling regular leader replacement at no additional communication cost or drop in system 
throughput. 
\HSTwo{}~\cite{hs2} achieves two-phase latency while maintaining linearity; the published \HSTwo{} algorithm is not streamlined, 
and streamlined \HSOne{} contributes a streamlined variant (as well as \early). 
Several other protocols have aimed for two-phase streamlined and linear latency.
However, Fast-HotStuff~\cite{fasths} and Jolteon~\cite{jolteon} have quadratic complexity in view-change; 
AAR~\cite{aar} employs expensive zero-knowledge proofs; Wendy~\cite{nocommit} relies  
on a new aggregate signature construction (and is super-linear); Marlin~\cite{marlin} introduces an additional \emph{virtual block}, offering leaders one more chance to propose a block extending the highest certificate that is supported by all correct replicas.

\textbf{Parallel Dissemination.}
Slotting is complementary to the prior multi-leader protocols like RCC~\cite{rcc,suyash-thesis}, MirBFT~\cite{mirbft}, and SpotLess~\cite{spotless}. 
These protocols focus mostly on increasing {\em throughput}, and a majority of them have a \HS-core. Thus, their designs are orthogonal to this paper. Any reduction in latency, the elimination of leader slowness phenomena, and tail-forking attacks will improve them.
Autobahn~\cite{autobahn} presents a data dissemination protocol that separates the task of disseminating client requests from consensus. It allows all replicas, in parallel, to batch and broadcast client requests. However, after dissemination, Autobahn employs PBFT to reach consensus on the execution order for all requests. Thus, Autobahn is orthogonal to the design of \HSOne{}; the PBFT consensus in Autobahn can be replaced with \HSOne{} to yield lower latency.
DAG-based consensus protocols~\cite{dagrider,narwhal,bullshark,shoal,cordial-miners,bbca-chain,fides,thunderbolt} decouple data dissemination from consensus by leveraging \emph{reliable broadcast (RBC)} mechanisms \cite{brachabroadcast}. These protocols construct a \emph{Directed Acyclic Graph (DAG)} of blocks generated by distinct replicas, enabling high throughput. However, this comes at the cost of increased latency introduced by RBC. Recent works~\cite{mysticeti,shoalpp,sailfish} have focused on reducing the latency of DAG-based consensus protocols. In this context, we posit that speculative execution offers a promising approach to further reduce latency.


\textbf{View Synchronization.}
The view-by-view paradigm of BFT protocols relies on view synchronization mechanisms to coordinate the replicas and to guarantee progress. 
Several solutions to the view synchronization problem have been proposed. 
Prior works~\cite{cogsworth,naor,boltdumbo,fastsync} have $O(n^3)$ worst-case message complexity. 
RareSync\cite{raresync} and Lewis-Pye~\cite{lewis} reduce the worst-case message complexity to $O(n^3)$ 
but face $O(n\Delta)$ latency in the presence of faulty leaders. Fever~\cite{fever} removes the $O(n\Delta)$ 
latency but assumes a synchronous start of replicas. Lumiere~\cite{lumiere} eliminates the need for the 
assumption and maintains all other properties of Fever. SpotLess~\cite{spotless} adopts a rapid view 
synchronization mechanism similar to FastSync~\cite{fastsync}, but embeds view synchronization into 
the \BFT{} consensus workflow, eliminating the need for a separate sub-protocol. 

{\bf Leader Slowness.}
The leader-slowness attack is a well-known problem in blockchains~\cite{time-role-mev,pbs2024,Daian2019FlashB2}.
Prior work has illustrated that in Ethereum, for $59\%$ of blocks, proposers have earned higher MEV rewards than block rewards~\cite{time-role-mev}, and 
any additional delay in proposing can help maximize their MEVs~\cite{time-is-money}.
There are two popular solutions to tackle leader slowness:
(i) Exclude any block that misses a set deadline to the main blockchain. However, a clever proposer can still delay proposing until the deadline~\cite{fork-based}.
(ii) Assign block rewards proportional to the number of attestations; a delayed block will receive fewer attestations and thus reduced block rewards~\cite{reduced-block-reward}. 
However, if MEV rewards exceed total block rewards, the proposer makes a profit despite losing any block reward.

{\bf Tail-forking attack.} 
As described earlier, BeeGees~\cite{beegees} describes the problem of tail-forking. 
They present an elegant solution to this problem by requiring replicas to store the proposal sent by the leader and forwarding that proposal in the future rounds. 
Unfortunately, resending these proposals over the network incurs additional bandwidth overhead.

{\bf Real-World Deployments.}
Several deployed blockchain systems, such as Espresso Systems HotShot~\cite{hotshot}, Flow Networks~\cite{flow-networks}, Meter~\cite{meter} have expressed a latency-over-everything emphasis.  
Early adopters of \HS{}, DiemBFT~\cite{diembft-hotstuff}, and Aptos that uses a two-phase variant of DiemBFT, Ditto~\cite{jolteon}, 
demonstrate the importance of latency. 
Recently, Spacecoin~\cite{spacecoin-bluepaper} unveiled plans to launch a trust platform operating within satellite-cubes in orbit, where latency is paramount because the link from Earth to satellites is slow.
All of these systems may benefit from incorporating \HSOne{}.

\section{Conclusion}
\label{sec:conclusion}
The principal goal of this work has been latency reduction for client finality confirmations in streamlined \BFT{} consensus protocols. We demonstrated that \HSOne{} successfully lowers latency algorithmically via speculation,
and furthermore, tackles leader-slowness and tail-forking attacks via slotting.
Additionally, we exposed and resolved the \textit{prefix speculation dilemma} that exists in the context of \BFT{} protocols that employ speculation.

\begin{acks}
This work is partially funded by NSF Award Number 2245373.
\end{acks}


\bibliographystyle{ACM-Reference-Format}
\bibliography{sources}

      \appendix
      \newpage
\section{Appendix}

\subsection{Speculation Safety in Basic-HotStuff-1} \label{app:safe-spec1}

Allowing replicas to speculatively execute transactions in a proposal $m$ upon receiving a certificate for $m$ is insufficient to guarantee safety for clients, i.e., a client mistakenly considers a transaction as committed after receiving $\n-\f$ responses for it.
The following examples demonstrate how speculative execution after observing a prepare certificate may violate safety unless both the \textbf{Prefix Speculation Rule} and the \textbf{No Gap Rule} are strictly followed.

\textbf{Prefix Speculation Rule}: We first present a scenario where the Prefix Speculation Rule is violated.
Assume the highest certificate across all replicas is $\Certificate{0}$, and the total number of replicas is $\n = 3\f+1$.
Partition the $2\f+1$ correct replicas into three disjoint sets: $A$, $A'$, and $A^*$, such that $\abs{A} = \abs{A'} = \f$ and $\abs{A^*} = 1$.
Suppose the first four leaders are Byzantine.
\begin{itemize}[wide]
    \item In view $1$, the leader $\Primary{1}$ proposes block $B_1$ extending $\Certificate{0}$. A quorum of $\n-\f$ replicas supports this proposal by sending threshold signature shares for $B_1$, allowing $\Primary{1}$ to form the prepare certificate $\Certificate{1}$. However, $\Primary{1}$ forwards $\Certificate{1}$ only to the $\f$ replicas in set $A$, who speculatively execute $B_1$ and respond to the client.
    
    \item In view $2$, leader $\Primary{2}$ disregards $\Certificate{1}$ and instead proposes a new block $B_2$ extending $\Certificate{0}$ to all replicas. Replicas in $A'$ and $A^*$ support $B_2$, allowing $\Primary{2}$ to form $\Certificate{2}$, which is forwarded only to $A'$. The $A'$ replicas then speculatively execute $B_2$ and respond to clients.
    
    \item In view $3$, $\Primary{3}$ ignores $\Certificate{2}$ and proposes $B_3$ extending $\Certificate{1}$ to all replicas. Replicas in $A$ and $A^*$ support $B_3$, enabling the formation of $\Certificate{3}$, which is sent only to $A^*$. Upon receiving it, $A^*$ replicas speculatively execute both $B_3$ and its ancestor $B_1$ (not following the Prefix Speculation Rule).
    
    \item In view $4$, $\Primary{4}$ disregards $\Certificate{3}$ and proposes $B_4$ extending $\Certificate{2}$ to all replicas. Replicas in $A$ and $A'$ support $B_4$, leading to the formation of $\Certificate{4}$. Although $\Certificate{2}$ conflicts with the highest known certificate $\Certificate{1}$ known to replicas in $A$, they are required to support $B_4$ due to the higher view number of $\Certificate{2}$. $\Primary{4}$ then broadcasts $\Certificate{4}$ to all replicas.

    Ultimately, $B_4$ becomes the highest known certificate across all replicas and will eventually be committed.
    
    \item However, an unsafe scenario for clients arises: the client for transactions in $B_1$ may have received $\n - \f$ responses from replicas in $A$, $A^*$, and $\f$ faulty replicas, even though $B_1$ will not be committed.
\end{itemize}

This example illustrates the \strongspec{} dilemma: replicas vote to commit a block $B_v$ along with its prefix, but cannot safely speculate on the prefix unless specific conditions are met.
According to the \strongspec{} rule (Definition~\ref{def:strongspec}), speculative execution is safe only when the prefix of $B_v$ is already committed. In this example, replicas in $A^*$ vote to commit a block $B_3$ along with its prefix $B_1$, but also speculate on the prefix $B_1$, violating the \strongspec{} rule. Thus, the client forms a commit-vote quorum for $B_1$ consisting of $A$, $A^*$, and $\f$ faulty replicas. However, the replicas in $A$, which voted for $B_1$ in view $1$, will switch to support a higher conflicting block $B_2$ after receiving $\Certificate{2}$ in view $4$, which makes the quorum invalid.

From prior literature on speculative consensus protocols, we note that Zyzzyva adopts a useful approach: replicas attach a view number to their speculative results, and clients are required not to aggregate responses from different views. This practice can help mitigate the risk of inconsistency caused by speculative execution across views.

\newpage
\textbf{No Gap Rule}: Secondly, we present a scenario where the No Gap Rule is violated, with the same assumption as we had in the previous scenario.

\begin{itemize}[wide]
    \item In view $1$, the leader $\Primary{1}$ proposes block $B_1$ extending $\Certificate{0}$. A quorum of $\n-\f$ replicas supports this proposal by sending threshold signature shares for $B_1$, allowing $\Primary{1}$ to form the prepare certificate $\Certificate{1}$. However, $\Primary{1}$ forwards $\Certificate{1}$ only to the $\f$ replicas in set $A$, who speculatively execute $B_1$ and respond to the client.
    
    \item In view $2$, leader $\Primary{2}$ disregards $\Certificate{1}$ and instead proposes a new block $B_2$ extending $\Certificate{0}$ to all replicas. Replicas in $A'$ and $A^*$ support $B_2$, allowing $\Primary{2}$ to form $\Certificate{2}$, which is forwarded only to $A'$. The $A'$ replicas then speculatively execute $B_2$ and respond to clients.

    \item In view $3$, $\Primary{3}$ ignores $\Certificate{2}$ and proposes $B_3$ extending $\Certificate{1}$ to $A^*$ only. Upon receiving $B_3$ extending $\Certificate{1}$, $A^*$ replicas  speculatively execute $B_1$ (not following the No Gap Rule). $\Primary{3}$ does not collect votes and then no certificate is formed.

    \item In view $4$, $\Primary{4}$ proposes $B_4$ extending $\Certificate{2}$ to all replicas. All replicas support $B_4$ because it extends the highest certificate $\Certificate{2}$. $\Primary{4}$ then broadcasts $\Certificate{4}$ to all replicas.
    Ultimately, $B_4$ becomes the highest known certificate across all replicas and will eventually be committed.
    
    \item However, an unsafe scenario for client arises: the client for transactions in $B_1$ may have received $\n - \f$ responses from replicas in $A$, $A^*$, and $\f$ faulty replica, even though $B_1$ will not be committed.
\end{itemize}

This example highlights the critical requirement that when a replica $R$ wishes to speculate on a block $B_v$, it must ensure there is \textbf{no view gap} between the view in which the prepare-certificate was formed and its current view. This is necessary to prevent speculative execution on a proposal that may be superseded by a higher certificate formed during the gap—one that remains unknown to $R$. In this example, for replica $A^*$ in view $3$, a view gap exists between its current view and the view in which $\Primary{1}$ formed its certificate; meanwhile, a higher certificate $\Certificate{2}$, formed by $\Primary{2}$ in the gap, will later supersede it.

According to the No Gap Rule (Definition~\ref{def:nogap}), in \BHSOne{}, if $R$ wishes to speculatively execute a block $B_w$, it is safe only if $w=v$ and $\Certificate{w}$ is formed in view $v$.

In summary, the general intuition behind both the \strongspec{} Rule and the No-Gap Rule is the same: \emph{a replica should not speculate on a block if there may exist a higher conflicting certificate that could supersede it}. The \strongspec{} Rule emphasizes this principle for blocks in the uncommitted prefix, while the No-Gap Rule focuses on the block of the latest received certificate.

With the two rules in place, receiving $\n - \f$ responses implies that at least $\f + 1$ correct replicas have speculatively executed the block. This, in turn, implies that $\f + 1$ correct replicas are locked on the certificate of the speculated block, which ensures that no higher conflicting certificate can be formed, thereby guaranteeing safe speculation.

\subsection{Rollback is Necessary} \label{app:rollbackscenes}

Providing \early{} responses is inherently speculative. If a conflicting certificate is later formed at a higher view, replicas must roll back their local-ledger state to maintain safety. We illustrate this necessity with the following scenario.

Assume the system starts in the initial state $\bot$. In view $1$, the leader $\Primary{1}$ proposes a message $B_1$ that extends $\Certificate{\bot}$. A quorum of $\n - \f$ replicas supports the proposal by sending threshold signature shares, allowing $\Primary{1}$ to form the prepare certificate $\Certificate{1}$. This certificate is forwarded to a subset of $\f$ correct replicas, denoted by set $A$. The replicas in $A$ speculatively execute the transactions in $B_1$ and respond to the client.

Now suppose the leader of view $2$, $\Primary{2}$, is also faulty and ignores the highest known certificate $\Certificate{1}$. It proposes a conflicting message $B_2$ extending $\Certificate{\bot}$ and broadcasts it to all replicas. A distinct set of $\n-\f$ replicas, disjoint from $A$, support $B_2$, allowing $\Primary{2}$ to form the conflicting certificate $\Certificate{2}$. $\Primary{2}$ then broadcasts $\Certificate{2}$ to all replicas.

Upon receiving $\Certificate{2}$, the replicas in set $A$ detect that it was formed at a higher view than $\Certificate{1}$, and consequently roll back their local ledger state. After the rollback, all correct replicas speculatively execute transactions in $B_2$ and respond to the client. Once the client receives responses from $\n - \f$ replicas, the transactions in $B_2$ are considered committed by the clients.

\subsection{Speculation Safety in Streamlined-HotStuff-1}
\label{app:safe-spec2}

The following examples demonstrate, in \PHSOne{}, how speculative execution after observing a prepare certificate may violate safety unless both the \textbf{Prefix Speculation Rule} and the \textbf{No Gap Rule} are strictly followed.

\textbf{Prefix Speculation Rule}: We first present a scenario where the Prefix Speculation Rule is violated.
Assume the initial highest certificate across all replicas is $\Certificate{0}$, and the total number of replicas is $\n = 3\f+1$.
Partition the $2\f+1$ correct replicas into three disjoint sets: $A$, $A'$, and $A^*$, such that $\abs{A} = \abs{A'} = \f$ and $\abs{A^*} = 1$.
Suppose the first eight leaders are Byzantine.

\begin{itemize}[wide]
    \item In view $1$, the leader $\Primary{1}$ proposes block $B_1$ that extends $\Certificate{0}$. A set of $\n - \f$ replicas support this proposal by sending their threshold signature shares for $B_1$ to the leader of view $2$, $\Primary{2}$, enabling it to form the prepare certificate $\Certificate{1}$. Assume that $\Primary{2}$ proposes block $B_2$ extending $\Certificate{1}$, but only forwards this certificate to a subset $A$ of $\f$ correct replicas. These replicas in $A$ speculatively execute $B_1$ and reply to the client.
    
    \item In view $3$, $\Primary{3}$ proposes block $B_3$ extending $\Certificate{0}$ and sends it to replicas in sets $A'$ and $A^*$. These replicas support $B_3$, enabling the leader of view $4$, $\Primary{4}$, to form a certificate $\Certificate{3}$. Assume $\Primary{4}$ then proposes block $B_4$ extending $\Certificate{3}$ and forwards it to set $A'$. The replicas in $A'$ speculatively execute $B_3$ and respond to the client.
    
    \item In view $5$, the leader $\Primary{5}$ ignores the higher certificate $\Certificate{3}$ and proposes $B_5$, which extends the lower certificate $\Certificate{1}$, to all replicas. Sets $A$ and $A^*$ support $B_5$, allowing $\Primary{6}$ (view $6$) to form a new certificate $\Certificate{5}$. $\Primary{6}$ forwards $\Certificate{5}$ only to set $A^*$, whose replicas speculatively execute $B_5$ and its prefix $B_1$ (not following the Prefix Speculation Rule).
    
    \item In view $7$, $\Primary{7}$ disregards the highest known certificate $\Certificate{5}$ and proposes $B_7$ extending $\Certificate{3}$ to all replicas. Sets $A$ and $A'$ support $B_7$, enabling $\Primary{8}$ to form certificate $\Certificate{7}$. Although $\Certificate{3}$ conflicts with the highest known certificate $\Certificate{1}$ known to replicas in $A$, they must support $B_7$ due to the higher view number of $\Certificate{3}$. $\Primary{8}$ broadcasts $\Certificate{7}$ to all replicas. 
    
    Eventually, $\Certificate{7}$ becomes the highest known certificate of all correct replicas and will be committed.
    
    \item However, an unsafe scenario for clients arises: the client for transactions in $B_1$ may have received $\n - \f$ responses from replicas in $A$, $A^*$, and $\f$ faulty replicas, even though $B_1$ will never be committed.
\end{itemize}

This example illustrates the \strongspec{} dilemma: replicas vote to commit a block $B_v$ along with its prefix, but cannot safely speculate on the prefix unless specific conditions are met.
According to the \strongspec{} rule (Definition~\ref{def:strongspec}), speculative execution is safe only when the prefix of $B_v$ is already committed. In this example, replicas in $A'$ vote to commit a block $B_5$ along with its prefix $B_1$, but also speculate on the prefix $B_1$, violating the \strongspec{} rule. Thus, the client forms a commit-vote quorum for $B_1$ consisting of $A$, $A^*$, and $\f$ faulty replicas. However, the replicas in $A$, which voted for $B_1$ in view $2$, will switch to support a higher conflicting block $B_3$ after receiving $\Certificate{3}$ in view $7$, which makes the quorum invalid.

\newpage
\textbf{No Gap Rule}: Secondly, we present a scenario where the No Gap Rule is violated, with the same assumption as we had in the previous scenario.

\begin{itemize}[wide]
    \item In view $1$, the leader $\Primary{1}$ proposes block $B_1$ that extends $\Certificate{0}$. A set of $\n - \f$ replicas support this proposal by sending their threshold signature shares for $B_1$ to the leader of view $2$, $\Primary{2}$, enabling it to form the prepare certificate $\Certificate{1}$. Assume that $\Primary{2}$ proposes block $B_2$ extending $\Certificate{1}$, but only forwards this certificate to a subset $A$ of $\f$ correct replicas. These replicas in $A$ speculatively execute $B_1$ and reply to the client.
    
    \item In view $3$, $\Primary{3}$ proposes block $B_3$ extending $\Certificate{0}$ and sends it to replicas in sets $A'$ and $A^*$. These replicas support $B_3$, enabling the leader of view $4$, $\Primary{4}$, to form a certificate $\Certificate{3}$. Assume $\Primary{4}$ then proposes block $B_4$ extending $\Certificate{3}$ and forwards it to set $A'$. The replicas in $A'$ speculatively execute $B_3$ and respond to the client.
    
    \item In view $5$, $\Primary{5}$ proposes $B_5$, which extends $\Certificate{1}$, to only set $A^*$, whose replicas speculatively execute $B_1$ (not following the No Gap Rule).
    
    \item In view $6$, $\Primary{6}$ proposes $B_6$ extending $\Certificate{3}$ to all replicas. All replicas support $B_7$, enabling $\Primary{7}$ to form certificate $\Certificate{6}$. Although $\Certificate{3}$ conflicts with the highest known certificate $\Certificate{1}$ known to replicas in $A$ and $A^*$, they must support $B_6$ due to the higher view number of $\Certificate{3}$. $\Primary{7}$ broadcasts $\Certificate{6}$ to all replicas.

    Eventually, $\Certificate{7}$ becomes the highest known certificate across all correct replicas and will be committed.
    
    \item However, an unsafe scenario for clients arises: the client for transactions in $B_1$ may have received $\n - \f$ responses from replicas in $A$, $A^*$, and $\f$ faulty replicas, even though $B_1$ will never be committed.
\end{itemize}

This example highlights the critical requirement that when a replica $R$ wishes to speculate on a block $B_v$, it must ensure there is \textbf{no view gap} between the view in which the prepare-certificate was formed and its current view. This is necessary to prevent speculative execution on a proposal that may be superseded by a higher certificate formed during the gap—one that remains unknown to $R$. In this example, for replica $A^*$ in view $5$, a view gap exists between its current view and the view in which $\Primary{1}$ formed its certificate; meanwhile, a higher certificate $\Certificate{3}$, formed by $\Primary{4}$ in the gap, will later supersede it.

According to the No Gap Rule (Definition~\ref{def:nogap}), in \PHSOne{}, if $R$ wishes to speculatively execute a block $B_w$, it is safe only if $w=v-1$ and $\Certificate{w}$ is formed in view $v$.

      \section{Correctness Proofs}
\label{app:proof}
In this Section, we prove the \emph{safety} and \emph{liveness} of \PHSOne{}. 
We first prove the {\em safety} guarantee.

\begin{lemma}\label{prop:non_divergent} Let $R_1$ and $R_2$ be two correct replicas that execute blocks $B_{v}^{1}$ and $B_{v}^{2}$ of view $v$. If $\n = 3\f + 1$, then $B_{v}^{1} = B_{v}^{2}$. \end{lemma}

\begin{proof} A correct replica $R_i$ executes a block $B_{v}^{i}$ only after obtaining a prepare certificate for $B_{v}^{i}$, as specified in Figure~\ref{alg:steamlined-hs1}, which consists of threshold signature shares from $\n - \f$ replicas.

Let $S_i$ denote the set of replicas that voted for the proposal containing $B_{v}^{i}$, so $\abs{S_i} = \n - \f = 2\f + 1$.
Let $X_i = S_i \setminus \f$ represent the subset of correct replicas in $S_i$.
Since at most $\f$ replicas may be faulty, we have $\abs{X_i} \geq 2\f + 1 - \f = \f + 1$.

Assume, for the sake of contradiction, that $B_{v}^{1} \neq B_{v}^{2}$.
This implies that $X_1 \cap X_2 = \emptyset$, since any common correct replica would not vote for two distinct blocks in the same view.
Therefore, the combined set $X_1 \cup X_2$ would contain at least $2(\f + 1) = 2\f + 2$ correct replicas.

However, this contradicts the total number of correct replicas in the system, which is $\n - \f = 2\f + 1$.
Thus, our assumption must be false, and we conclude that $B_{v}^{1} = B_{v}^{2}$. \end{proof}

\begin{lemma} \label{lm:no-gap-rule}
    If a replica $R$ receives a certificate $\Certificate{v+1}$ that extends certificate $\Certificate{v}$, then for any view $w > v$, no certificate $\Certificate{w}$ that conflicts with $\Certificate{v}$ can exist.
\end{lemma}

\begin{proof}
    We know that a replica $R$ received $\Certificate{v+1}$ that extends $\Certificate{v}$, 
    which is only possible if $\n-\f = 2\f+1$ replicas that set $\Certificate{v}$ as their higher known certificate also voted for $\Certificate{v+1}$.
    Let's denote the $\f+1$ correct replicas from these $\n-\f$ replicas as $A$.
    Further, certificate $\Certificate{w}$ conflicts with $\Certificate{v}$, $w > v$, which 
    implies that $\Certificate{v}$ and $\Certificate{w}$ extend the same ancestor and $\Certificate{w}$ received support of $\n-\f = 2\f+1$ replicas.
    Let's denote the $\f+1$ correct replicas from these $\n-\f$ replicas as $A'$.
    As $w \neq v+1$, so $w > v+1$.
    Moreover, any correct replica that sets $\Certificate{v}$ as its highest known certificate will not vote for a conflicting block.
    Thus, $\abs{A}+\abs{A'} = 2\f+2$, which is more than the total number of correct replicas and is a contradiction. 
\end{proof}

\begin{corollary} \label{corollary:uniquebranch} If $\f + 1$ correct replicas speculatively execute a block $B_v$, then no higher-view block $B_w$ with $w > v$ that conflicts with $B_v$ can be committed. \end{corollary}

\begin{proof} By Lemma~\ref{lm:no-gap-rule}, if a correct replica speculatively executes a block $B_v$, it must have observed a valid certificate $\Certificate{v}$ for $B_v$ and set it as its highest known certificate.

If $\f + 1$ correct replicas have speculatively executed $B_v$, then at least $\f + 1$ correct replicas have locked on $\Certificate{v}$. Since there are only $2\f + 1$ correct replicas in total, no quorum of $\n - \f = 2\f + 1$ votes can be collected for any conflicting block $B_w$ with $w > v$, as at least $\f + 1$ correct replicas will refuse to vote for any certificate conflicting with $B_v$.

Thus, no conflicting certificate can be formed at a higher view, and therefore no conflicting block can be committed. \end{proof}

\begin{lemma}\label{lm:certificate} If a correct replica $R$ commits a block $B_v$, then no conflicting block can be committed.\end{lemma}

\begin{proof} Assume, for contradiction, that there exists a block $B_w$ proposed in view $w > v$ that conflicts with $B_v$, and that another correct replica $R'$ has committed $B_w$. This would imply that the global ledgers at replicas $R$ and $R'$ have diverged, violating safety.

For both $B_v$ and $B_w$ to be committed, replicas $R$ and $R'$ must have followed the prefix commit rule described in \S\ref{ss:pipe-onephase}:
\begin{itemize} \item $R$ must have received a certificate $\Certificate{v+1}$ that extends $\Certificate{v}$, thereby committing $B_v$. \item $R'$ must have received a certificate $\Certificate{w+1}$ that extends $\Certificate{w}$, thereby committing $B_w$. \end{itemize}

Since $w > v$ and $w \ne v+1$, it follows that $w > v+1$.
However, by Lemma~\ref{lm:no-gap-rule}, once certificates $\Certificate{v}$ and $\Certificate{v+1}$ are formed, no conflicting certificate $\Certificate{w}$ for $w > v$ can be constructed.
Therefore, the assumption that $B_w$ was committed by $R'$ leads to a contradiction.

Hence, once a correct replica commits $B_v$, no conflicting block can be committed. \end{proof}

\begin{theorem}\label{thm:safety} \textbf{(Safety)} \PHSOne{} guarantees consensus safety in a system with $n \ge 3\f + 1$ replicas: if two correct replicas $R_1$ and $R_2$ commit blocks $B_1$ and $B_2$, respectively, at the same position $k$ in the global ledger, then $B_1 = B_2$. \end{theorem}

\begin{proof} By Lemma~\ref{lm:certificate}, once a correct replica commits a block, no conflicting block can be committed. Therefore, both $B_1$ and $B_2$ are permanently part of the respective global ledgers of $R_1$ and $R_2$, and must not conflict.

Suppose, for the sake of contradiction, that $B_1 \ne B_2$ and that $B_2$ extends $B_1$. Then, $B_1$ must occupy position $k$ in the global ledger of $R_1$, and also appear as part of the prefix of $B_2$. This would imply that the prefix of $B_2$ contains $k-1$ blocks (including $B_1$), while the prefix of $B_1$ contains at most $k-2$ blocks. However, this contradicts the assumption that $B_1$ is committed at position $k$ with $k-1$ blocks in its prefix.

Consequently, $B_1$ and $B_2$ must be of the same view and must be the same block, and we conclude that no two correct replicas can commit different blocks at the same position. Hence, \PHSOne{} guarantees consensus safety. \end{proof}

Next, we prove the liveness guarantee of streamlined \HSOne{}.
Following prior works~\cite{hotstuff,hs2}, we assume the existence of a Global Stabilization Time (GST) and an appropriately chosen view timer length $\tau$, such that correct replicas eventually overlap in the same view after view synchronization.
This assumption ensures that the timer is long enough for the leader to process \MName{NewView} messages, obtain the highest known certificate, propose a block, and for replicas to respond with votes.
We denote by $v_s$ the first synchronized view after GST.

\begin{lemma}\label{lm:supportedbyall}
    For the leader $\Primary{v}$ of view $v$, where $v \ge v_s$, its proposal will be supported by all correct replicas.
\end{lemma}

\begin{proof}
    Assume that $\Primary{v}$ enters view $v$ at time $t$. According to~\cite{raresync,lewis}, the PaceMaker ensures that all correct replicas enter view $v$ by $t + 2\Delta$. Thus, $\Primary{v}$ can receive \MName{NewView} messages from all correct replicas by $t + 3\Delta$. If $\Primary{v}$ forms a certificate $\Certificate{v-1}$, then it is guaranteed to be the highest certificate known to any correct replica; otherwise, the highest certificate can still be learned from the received \MName{NewView} messages.

    The proposal sent by $\Primary{v}$ will arrive at all correct replicas by $t + 4\Delta$. By setting a sufficiently long timer, all correct replicas remain in view $v$ upon receiving the proposal and will vote for it.
\end{proof}

\begin{lemma}\label{lm:3consecutive}
     Assume three consecutive correct leaders: $\Primary{v}$, $\Primary{v+1}$, and $\Primary{v+2}$, with $v \ge v_s$. If $\Primary{v}$ proposes a block $B_{v}$ in view $v$, then all correct replicas will commit $B_{v}$ in view $v+2$.
\end{lemma}

\begin{proof}
    By Lemma~\ref{lm:supportedbyall}, in view $v$, block $B_{v}$ will be supported by all correct replicas, allowing $\Primary{v+1}$ to form a certificate $\Certificate{v}$.

    Similarly, in view $v+1$, block $B_{v+1}$ extending $\Certificate{v}$ will be supported by all correct replicas, enabling $\Primary{v+2}$ to form $\Certificate{v+1}$.

    Then, in view $v+2$, all correct replicas will receive block $B_{v+2}$ extending $\Certificate{v+1}$, which satisfies the prefix commit rule for $B_{v}$. Hence, $B_{v}$ will be committed by all correct replicas in view $v+3$.
\end{proof}

\begin{theorem}\label{thm:liveness}
    (\textbf{Liveness}) All correct replicas eventually commit a transaction $T$.
\end{theorem}

\begin{proof}
    Since there are $\n = 3\f + 1$ replicas and \HSOne{} rotates leaders in a round-robin manner, there must exist a set of three consecutive correct leaders: $\Primary{v}$, $\Primary{v+1}$, and $\Primary{v+2}$, with $v \ge v_s$. By Lemma~\ref{lm:3consecutive}, any transaction $T$ contained in the block proposed in view $v$ will eventually be committed by all correct replicas.
\end{proof}

\begin{corollary}\label{corl:2consecutive}
    Assume two consecutive correct leaders: $\Primary{v}$ and $\Primary{v+1}$, with $v \ge v_s$. If $\Primary{v}$ proposes a block $B_{v}$ in view $v$, then $B_{v}$ will eventually be committed.
\end{corollary}
     
\begin{proof}
    Consider the setting in Lemma~\ref{lm:3consecutive}, stopping at two consecutive correct leaders: $\Primary{v}$ and $\Primary{v+1}$. In view $v+2$, all correct replicas will eventually receive a block $B_{v+1}$ that extends $\Certificate{v}$ and will adopt $\Certificate{v}$ as their highest known certificate. This ensures that in any future view, no conflicting certificate with $\Certificate{v}$ can be formed. 
    
    By Theorem~\ref{thm:liveness}, we know that there will eventually be a set of three consecutive correct leaders—say $\Primary{w}$, $\Primary{w+1}$, and $\Primary{w+2}$, with $w \ge v+1$—who commit a block $B_{w}$ that extends the chain containing $B_{v}$, because $B_{v}$ is a lower-view non-conflicting block of $B_{w}$. Since $B_{v}$ is in the prefix of that chain, all correct replicas will eventually commit $B_{v}$.
\end{proof}

\begin{corollary} \label{corl:noconflictcommit} (\textbf{Client Safety}) If a client receives $\n - \f$ matching responses for transactions in a block $B_v$, then $B_v$ will eventually be committed by all correct replicas. \end{corollary}

\begin{proof} By Lemmas~\ref{lm:no-gap-rule} and~\ref{lm:certificate}, we derive this result as follows:

If a client receives $\n - \f$ responses for transactions in block $B_v$, then at least $\f + 1$ of these responses must have come from correct replicas. There are two cases to consider:

\begin{enumerate} \item \textbf{Speculative execution:} At least $\n - \f$ replicas speculatively executed $B_v$ and replied to the client. This quorum must include at least $\f + 1$ correct replicas. By Corollary~\ref{corollary:uniquebranch}, this implies that no higher-view certificate conflicting with $B_v$ can be formed, and thus no conflicting block can be committed.
\item \textbf{Committed execution:} At least one correct replica committed and executed $B_v$. By Lemma~\ref{lm:certificate}, no conflicting block can be committed.

\end{enumerate}

In either case, once the client receives $\n - \f$ matching responses for $B_v$, it is guaranteed that no conflicting block can later be committed. 

By Theorem~\ref{thm:liveness}, we know that there will eventually be a set of three consecutive correct leaders—say $\Primary{w}$, $\Primary{w+1}$, and $\Primary{w+2}$, with $w \ge v+1$—who commit a block $B_{w}$ that extends the chain containing $B_{v}$, because $B_{v}$ is a lower-view non-conflicting block of $B_{w}$. Since $B_{v}$ is in the prefix of that chain, all correct replicas will eventually commit $B_{v}$.
\end{proof}


\end{document}